\documentclass[journal]{IEEEtran}

\usepackage{amsfonts,amsmath,amssymb}
\usepackage{dsfont}
\usepackage{cite}
\usepackage{graphicx}
 \usepackage{booktabs}
\usepackage{url}
\usepackage{gensymb}
\usepackage{caption}
\usepackage{bm}
\usepackage{bbding}
\usepackage{makecell}
\usepackage{amsmath}
\usepackage{enumerate}
\usepackage{multirow}
\usepackage[ruled,linesnumbered]{algorithm2e}
\usepackage{algorithmicx}
\usepackage{xcolor}
\usepackage{algpseudocode}
\usepackage{algcompatible}
\usepackage{subcaption}
\usepackage{booktabs}
\usepackage{array}
\usepackage{slashbox}
\usepackage{CJK}
\usepackage{esint} 

\usepackage{graphicx}
\usepackage{amsthm}

\newtheorem{lemma}{Lemma}
\newtheorem{proposition}{Proposition}

\newtheorem{remark}{Remark}

\newcommand{\figref}[1]{\figurename~\ref{#1}}

\DeclareMathOperator*{\argmin}{arg\,min}
\DeclareMathOperator{\diag}{diag}

\begin{document}
\title{Amplitude-Constrained Constellation and Reflection Pattern Designs for Directional Backscatter Communications Using Programmable Metasurface  }
\author{Wei~Wang,~\IEEEmembership{Member,~IEEE}, Bincheng~Zhu,~\IEEEmembership{Senior Member,~IEEE}, Yongming~Huang,~\IEEEmembership{Senior Member,~IEEE}, and~Wei~Zhang,~\IEEEmembership{Fellow,~IEEE}
\thanks{

W. Wang is with Peng Cheng Laboratory, Shenzhen 518066, China (e-mail: wei\_wang@ieee.org).

B. Zhu and Y. Huang are with School of Information Science and Engineering, Southeast University, Nanjing 210096, China, and also with the Purple Mountain Laboratory, Nanjing 211111, China (e-mail:zbc@seu.edu.cn; huangym@seu.edu.cn).

W. Zhang is with School of Electrical Engineering and Telecommunications, The University of New South Wales, Sydney, NSW 2052, Australia (e-mail: w.zhang@unsw.edu.au).

}
}
\maketitle

\begin{abstract}
The large scale reflector array of programmable metasurfaces is capable of increasing the power efficiency of backscatter communications via passive beamforming and thus has the potential to revolutionize the low-data-rate nature of backscatter communications. In this paper, we propose to design the power-efficient higher-order constellation and reflection pattern under the amplitude constraint brought by backscatter communications. For the constellation design, we adopt the amplitude and phase-shift keying (APSK) constellation and optimize the parameters of APSK such as ring number, ring radius, and inter-ring phase difference. Specifically, we derive closed-form solutions to the optimal ring radius and inter-ring phase difference for an arbitrary modulation order in the decomposed subproblems. For the reflection pattern design, we propose to optimize the passive beamforming vector by solving a multi-objective optimization problem that  maximizes reflection power and guarantees beam homogenization within the interested angle range. To solve the problem, we propose a constant-modulus power iteration method, which is proven to be monotonically increasing, to maximize the objective function in each iteration. Numerical results show that the proposed APSK constellation design and reflection pattern design outperform the existing modulation and beam pattern designs in programmable metasurface enabled backscatter communications.
\end{abstract}


\section{Introduction}

The rapid growth of Internet of Things (IoT), driven by the development of ubiquitous computing, commodity sensors, and 5G mobile communications, is envisioned to forge a technological path into smart cities for human beings. A major challenge of IoT is the design of  energy-efficient and low-hardware-cost communication module \cite{liu2019next}. Backscatter communication is a communication technique that allows wireless system to transmit information without the aid of bulky and power-hungry radio frequency (RF) components on the transmitter \cite{BackScattTCOM}, which offers a solution to low-cost and energy-efficient wireless communications. Hence, backscatter communication is widely used in short range communication scenarios such as radio-frequency identification (RFID), and IoT sensors.

The recent progress of programmable metasurface, which is characterized by the capacity of tailoring electromagnetic waves, is revolutionizing the design paradigm of wireless communications  \cite{Tutorial, cui2010metamaterials,  yang2016programmable, wang2021joint, wang2022intelligent, zhang2021intelligent, DistributedIRS, vardakis2021intelligently, TangTWC, f2020perfect, hu2022irs}. In \cite{wang2021joint, wang2022intelligent, zhang2021intelligent, DistributedIRS}, the programmable metasurface, a.k.a. intelligent reflecting surface (IRS)/reconfigurable intelligent surface (RIS), is regarded as part of the wireless channel, which empowers human beings to proactively change radio propagation conditions.
Given that the reflecting element of IRS can be interpreted as an impedance-modulated antenna,  \cite{vardakis2021intelligently} shows that commercial RFID tags can be used as building blocks for a wirelessly-controlled and battery-free IRS implementation. In \cite{TangTWC}, free-space path loss models for programmable metasurfaces-assisted wireless communications, which  underpin the theoretical analysis of wireless performance boost from RIS, are developed and then corroborated both analytically and experimentally.
In \cite{f2020perfect}, a massive backscatter communication scheme based on the extreme sensitivity of the perfect absorption condition is implemented with a programmable metasurface in a rich-scattering environment to achieve physical layer security. In \cite{hu2022irs}, distributed semi-passive programmable metasurfaces are deployed to the design of an integrated sensing and communication system.

Other than being part of the radio propagation environment to assist wireless communication, programmable metasurfaces also play an important role in backscatter communications. The large-scale reflector array has enabled more diverse applications of backscatter communications by increasing power efficiency via passive beamforming. In  \cite{zhao2020metasurface}, the pioneering work has been done to validate the feasibility of backscatter communications using  large-scale programmable metasurfaces,  which are characterized by the large aperture and huge degrees of freedom. Specifically, a  secure ambient backscatter communication system that leverages existing commodity 2.4 GHz Wi-Fi signals is designed and implemented, and the prototype  achieves  the data rate on the order of hundreds of Kbps. In \cite{tang2019wireless,tang2020mimo, chen2022accurate}, metasurface-based backscatter communications with a dedicated single-tone sinusoidal source are investigated. For backscatter communications powered by a feed antenna, the information is encoded by modulating the incident single-tone sinusoidal wave through varying the impedance that determines the reflection coefficient \cite{BackScattTCOM, lazaro2020feasibility}.  With a pre-designed impedance set, backscatter modulation can be realized by choosing the impedance according to the input binary bits.  In \cite{tang2019wireless}, the prototype of programmable metasurface-based backscatter communication using quadrature phase shift keying (QPSK) is implemented and evaluated. In \cite{tang2020mimo, chen2022accurate},  the non-linear harmonic designs of high-order QAM modulations and multiple input multiple output (MIMO)  data transmissions using programmable metasurfaces are proposed. The proposed harmonic control is effective in tackling the coupling effects of the reflection amplitude and phase responses of an unit cell \cite{abeywickrama2020intelligent}, which is a promising technique to enable decoupled and flexible amplitude and phase control of metasurfaces. On the basis of the non-linear harmonic designs, the advanced designs of high-order modulation and beamforming techniques have become feasible to be implemented in promgrammable metasurfaces. However, as conventional backscatter communication suffers from short transmission range and low data rate \cite{liu2019next,thomas2010qam}, the energy efficient design of higher-order backscatter modulations attracts very limited research interests. Leveraging the large-scale reflector array of programmable metasurfaces, passive beamforming will significantly improve the effective radiation power of backscatter communications \cite{LISAliang}. Thus, the constellation design of higher-order backscatter modulation becomes essential. On the other hand, although beam pattern design has been extensively investigated for conventional MIMO system under different antenna structures, e.g., fully digital MIMO \cite{long2019window,cheng2021analytical}, analog MIMO \cite{xiao2016hierarchical,zhu20193}, and hybrid digital and analog MIMO  \cite{alkhateeb2014channel, wang2020optimal}. The applicability of the precedent designs to programmable metasurface enabled passive beamforming remains unexplored. It is noteworthy that the design constraint of programmable metasurface enabled backscatter communications is inherently different from conventional MIMO system. Firstly, programmable metasurface reflects power rather than generates power, thereby the reflection coefficient vector does not comply to the sum power constraint of beamforming vector in conventional MIMO beamforming. Secondly,  constrained by the electromagnetic properties of the metamaterials, the reflectivity of the metasurface units is constrained \cite{Tutorial}. With the given incident signal emitted by the feed antenna, the amplitude of the reflected signal is upper bounded as a result of the constrained reflectivity.

To improve the power efficiency of directional backscatter communications and provide design guidelines for the practical backscatter communication systems\cite{tang2019wireless, zhao2020metasurface}, we propose to design the constellation and the reflection pattern under amplitude constraint. We firstly decompose the design of constellation and reflection pattern into two sub-problems. With respect to constellation design, we let the constellation follow the form of amplitude and phase-shift keying (APSK) and then propose to optimize the parameters of APSK under amplitude constraint. With respect to reflection pattern design, we firstly analyze the reflected power efficiency and then investigate the comparability of off-the-shelves beam pattern designs under sum power constraint with programmable metasurface enabled passive beamforming. Based on the obtained analytical results, we propose to optimize the reflection pattern under  constant modulus constraint, which is harsher than amplitude constrain. The main contributions we have made in this paper are summarized as follows:
\begin{itemize}
\item Following the criterion of maximizing minimum Euclidean distance, we optimize ring number, ring radius, and phase difference of inter-ring constellation points of APSK. Specifically, we derive closed-form  solutions to the optimal ring radius and inter-ring phase difference for an arbitrary modulation order. The generated APSK constellation by our algorithm is superior to conventional QAM and PSK with respect to the minimum Euclidean distance under amplitude constraint.
\item We analyze the reflected power efficiency of programmable metasurface under amplitude constraint, and our analysis results show that the sum power of the reflected signals in all directions is proportional to the squared $\ell_2$ norm of the passive beamforming vector, which indicates that the maximum sum reflected power is achieved if and only if the passive beamforming vector follows constant modulus constraint.
\item For reflection pattern design, we propose a multi-objective optimization problem to  maximize reflection power and guarantee beam homogenization within the interested angle range. Specifically, we formulate a max-min optimization problem  under constant modulus constraint. A constant-modulus power iteration method, which is proven to be monotonically increasing, is proposed to optimize the objective function in each iteration. Through analyzing the ripple factor and power ratio of the generated beam pattern, we validate the effectiveness of our proposed design.
\end{itemize}
\noindent Numerical results show that the proposed APSK constellation design and reflection pattern design outperform the existing modulation schemes (e.g., QAM modulation) and the beam pattern designs under the sum power constraint for programmable metasurface enabled backscatter communications.

The rest of the paper is organized as follows. Section II introduces the system model. In Section III, we perform APSK constellation design under amplitude constraint. In Section IV,  we perform reflection pattern design under constant modulus constraint. In  Section V,  numerical results are presented. Finally, in Section VI, we draw the conclusion.


{\em{Notations:\quad}} Column vectors (matrices) are denoted by bold-face lower (upper) case letters,  $(\cdot)^T$ and $(\cdot)^{H}$  represent transpose and  conjugate transpose operation, respectively,  ${ \rm gcd}(a, b)$ stands for the greatest common divisor of the integers $a$ and $b$, and  ${\rm lcm}(a, b)$  stands for the least common multiple of integers $a$ and $b$.

\section{System Model}

In this section, we introduce the system model of signal constellation in programmable metasurface enabled directional backscatter communications.

\subsection{Amplitude Constraint of The Reflected Signal}

\begin{figure}[tp]{
\begin{center}{\includegraphics[ height=5cm]{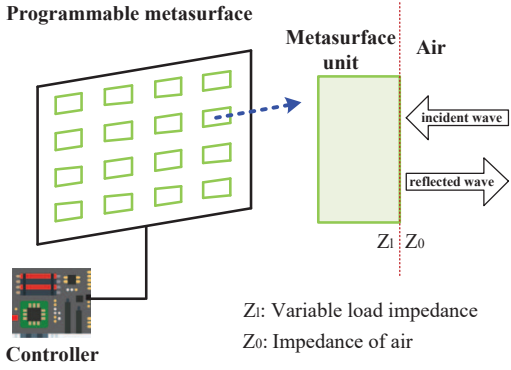}}
\caption{Illustration of programmable metasurface and its unit}\label{MetaUnit}
\end{center}}
\end{figure}

\figref{MetaUnit} shows the working mechanism of programmable metasurfaces. A dedicated incident signal $E_i(t)$ impinges on the metasurface, and the reflected signal ${E}_r(t)$ can be configured by the controller via tuning the load impedance. Following \cite{tang2019wireless, tang2020mimo},  the incident signal considered is emitted by a dedicated power source and is a single-tone sinusoidal\footnotemark, i.e., $E_i(t) = A \cos(2\pi f_c t + \varphi_0)$, and thus the reflected signal ${E}_r(t)$ is a single-tone carrier signal as well.

\footnotetext{In accordance with  \cite{tang2019wireless, tang2020mimo}, we assume that the programmable metasurface is placed in the far field of the horn antenna. Thus, the incident signal $E_i(t)$ is identical for all the metasurface units.
}

Applying phasor method in circuit analysis, we obtain the relationship between the incident signal  and the reflected signal  as follows
\begin{align}
{\mathbb{E}}_r  =   \mathbb{E}_i \cdot \Gamma(\omega)
=  A|\Gamma(\omega)| e^{j (\varphi_0 + \varphi_{\Gamma(\omega)})}  \label{RefEq}
\end{align}
where $\mathbb{E}_i  =  A e^{j\varphi_0 }$ is the phasor-domain representation (in exponential form) of the incident signal $E_i(t)$, $\mathbb{E}_r$ is the phasor-domain representation of the reflected signal $E_r(t)$, and  $\Gamma$ is the reflection coefficient that describes the fraction of the electromagnetic wave reflected by an impedance discontinuity in the transmission medium \cite{tang2019wireless,zhang2021wireless, hall2000high}, and $\Gamma(\omega) = |\Gamma(\omega)| e^{j \varphi_{\Gamma(\omega)}} $  denotes the exact value of the reflection coefficient when the angular frequency is $\omega = 2\pi f_c$. To be concise, $\omega$ will be omitted in the following context. According to \cite{zhu2013active, hall2000high, taflove2005computational, tang2019wireless}, the expression of the reflection coefficient  is given by
\begin{align}
\Gamma &= \frac{Z_l - Z_0}{Z_l + Z_0}
\end{align}
where $Z_l$ is the equivalent load impedance of metasurface unit, and $Z_0$ is the impedance of air. Through tuning $Z_l$,  the magnitude and phase of the reflection coefficient $\Gamma$ can be configured by the controller of programmable metasurface \cite{zhu2013active}, and the exact value of the load impedance  $Z_l$ can be derived via the well-known Smith chart.

Owing to the physical property of passive reflection, the reflection coefficient $\Gamma$ satisfies
\begin{align}
|\Gamma|\leq 1 \label{RefAmpCons}
\end{align}
 To explore the theoretical performance upper bound of the backscatter communications, we neglect the current hardware constraints \cite{abeywickrama2020intelligent, tang2020mimo} and  assume that the amplitude and the phase of the reflection coefficient can be tuned independently and continuously under the constraint \eqref{RefAmpCons}.
Thus, the magnitude of the reflected signal satisfies
\begin{align}
|{\mathbb{E}}_r | \leq A \label{AmpCons}
\end{align}

In the realm of telecommunications,  ${\mathbb{E}}_r$ is also referred to as equivalent baseband signal of the reflected signal ${E}_r(t)$. Thus, \eqref{AmpCons} indicates that the baseband signal of backscatter modulation is amplitude constrained, which is different from the power constraint imposed on the traditional communication systems.

\subsection{Backscatter Modulation}

Backscatter modulation is realized by changing the reflection coefficient according to the input information. In order to convey information, the reflection coefficient is configured as
\begin{align}
\Gamma(t) = \sum_{m=1}^M \Gamma[m] h(t - mT) \label{refcoeff}
\end{align}
where $\Gamma[m]$ is selected from a pre-designed finite alphabet, i.e., $\Gamma[m] \in \{I_1, I_2, \cdots, I_{|\mathcal{S}|} \} $, according to the input information bits, $h(t)$ is the rectangular pulse, and $T$ is symbol duration.

Thus, the constellation alphabet of the equivalent baseband signal ${\mathbb{E}}_r$  is given by
\begin{align}
\mathcal{S} = \{A e^{j  \varphi_0} I_1, A e^{j  \varphi_0}I_2, \cdots, A e^{j  \varphi_0}I_{|\mathcal{S}|}\}
\end{align}
According to \eqref{AmpCons}, the constellation of programmable metasurface enabled backscatter communications is under an amplitude constraint, i.e.,
\begin{align}
\max_{s \in \mathcal{S}} |s| \leq A \label{Const1}
\end{align}

Although the amplitude of the reflected signal can be changed by controlling the transmit power of the horn antenna to obey the conventional power constraint,  this structure requires (1) the strict synchronization between the programmable metasurface and the horn antenna, (2) the amplitude modulation capacity of the horn antenna, and (3) the linear wideband power amplifier for amplitude modulation at the horn antenna side. Also, the amplitude of all the metasurface units can only be controlled collectively. In a word, amplitude control at horn antenna side is not cost-effective and does not fully utilize the advantages of programmable metasurfaces. Thus, we follow the structure proposed by  \cite{tang2019wireless, tang2020mimo} in our work and let the horn antenna simply provide an incident pure carrier signal.

\subsection{Passive Beamforming}

A salient advantage of programmable metasurfaces over the conventional backscatter antennas is its capability of directional beamforming, which is  brought by the large number of metasurface units.  Concatenating the incident and reflected signal of the $N$ metasurface units, we have the vector-form representation of \eqref{RefEq} as
\begin{align}
{\mathbf{E}}_r  =   \mathbb{E}_i \cdot \mathbf{f}  \label{PassBeam}
\end{align}
where $\mathbf{f}$ is the reflection coefficient vector, and, according to \eqref{RefAmpCons}, $\mathbf{f}$  follows the amplitude constraint, i.e.,
\begin{align}
\mathbf{f}(i) \leq 1, \;\forall n \in \{1, .., N\} \label{ConstF}
\end{align}
When $\mathbf{f} = \mathbf{1}$, the reflection  pattern is omni-directional, which is apparently energy inefficient.   Since the intended users of backscatter communications are usually distributed in a constrained area, e.g., a plaza, a street block, a road section, the ideal reflection pattern should be concentrated and homogeneous within the interested area.

We assume that the metasurface units are arranged as an $N_x \times N_y$ rectangular planar array, and the array response vector of programmable metasurface is represented as
\begin{align}
\mathbf{v}(\Psi_x, \Psi_y) = \mathbf{v}(\Psi_x) \otimes  \mathbf{v}(\Psi_y) \label{SteeringVB}
\end{align}
 and
\begin{subequations}
\begin{align}
&\quad  \mathbf{v}(\Psi_x) =  \left[1,\; e^{j  \pi \Psi_x},\; \cdots,\; e^{j  (N_{x} - 1)\pi \Psi_x} \right]^T \label{UPA1}\\
 &\quad  \mathbf{v}(\Psi_y)  = \left[1,\; e^{j  \pi \Psi_y },\; \cdots,\; e^{j  (N_{y} - 1)\pi \Psi_y } \right]^T \label{UPA2}
\end{align}
\end{subequations}
where $\Psi_x  $ and $\Psi_y$ are cosine of the angle of arrival/depature (AoA/AoD), a.k.a. direction cosines, \cite{tsai2018millimeter, wang2021jittering}. To maximize the reflected signal power within the interested angle range and guarantee uniform signal strength, i.e., beam homogenization, in the meantime, we propose the following design criteria.

\noindent\emph{\textbf{ Criterion 1}}: Maximized reflected power over the intended angle range, namely
\begin{align}
   \max_{\mathbf{f}} P_{\mathcal{D}_{\Psi}}
\end{align}

\emph{\textbf{Criterion 2}}: Minimized the ripple factor within the intended angle range  $[\theta_{L},\; \theta_{U})$, namely
\begin{align}
\min_{\mathbf{f}}  \frac{V_{Ripple}}{V_{Mean}}
\end{align}
where
\begin{subequations}
\begin{align}
 & P_{\mathcal{D}_{\Psi}}  =  \oiint_{\mathcal{D}_{\Psi}}  \mathbf{f}^H \mathbf{v}(\Psi_x, \Psi_y) \mathbf{v}^H(\Psi_x, \Psi_y)\mathbf{f} \; d \Psi_x d \Psi_y  \label{Subeq1} \\
 & V_{Mean} = \frac{1}{\mathcal{A}(\mathcal{D}_{\Psi})}  \oiint_{\mathcal{D}_{\Psi}}  |\mathbf{v}^H(\Psi_x, \Psi_y)\mathbf{f}| \;  d \Psi_x d \Psi_y \label{Subeq3} \\
 &  V_{Ripple} \notag \\
& = \sqrt{\frac{1}{\mathcal{A}(\mathcal{D}_{\Psi})}  \oiint_{\mathcal{D}_{\Psi}} \left(|\mathbf{v}^H(\Psi_x, \Psi_y) \mathbf{f}| - V_{Mean} \right)^2 \;  d \Psi_x d \Psi_y} \label{Subeq4}
\end{align}
\end{subequations}
and $P_{\mathcal{D}_{\Psi}}$ is the power of the reflected signal within the intended angle range $\mathcal{D}_{\Psi}$,  $V_{Mean} $ is the mean voltage of the reflected signal  over the intended angle range  $\mathcal{D}_{\Psi}$,  $V_{Ripple}$ is root mean square (RMS) of the ripple voltage,  $\frac{V_{Ripple}}{V_{Mean}}$ is the ripple factor that measures the degree of fluctuations over the intended angle range  $\mathcal{D}_{\Psi}$, and $\mathcal{A}(\mathcal{D}_{\Psi})$ is the area of $\mathcal{D}_{\Psi}$.

\begin{figure}[tp]{
\begin{center}{\includegraphics[ height=4cm]{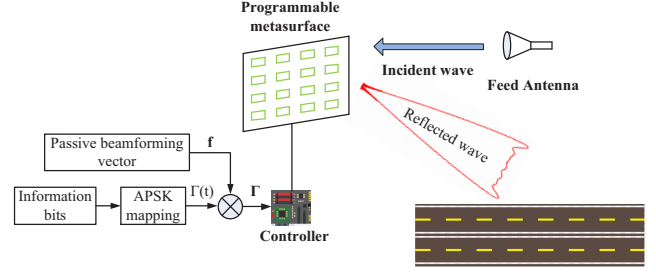}}
\caption{Programmable metasurface enabled directional backscatter modulator  with a road section as the intended area} \label{SystemModel}
\end{center}}
\end{figure}

\subsection{Directional Backscatter Communications}
\figref{SystemModel} shows the working mechanism of signal modulation in programmable metasurface enabled backscatter communications.  A single-tone carrier signal $E_i(t)$ impinges on the programmable metasurface from a feed antenna. Signal modulation is realized by collectively changing the reflection coefficients  of programmable metasurface according to the incoming information bits, and passive beamforming is realized by imposing beamforming vector on the metasurface units.
Thus, the reflection coefficient of directional backscatter communications can be represented as the product of the information-bearing factor $\Gamma(t)$ in \eqref{refcoeff} and the passive beamforming factor $\mathbf{f}$ (refer to \eqref{PassBeam}), i.e.,
\begin{align}
\boldsymbol{\Gamma} = \Gamma(t) \mathbf{f}
\end{align}
Note that $\Gamma(t)$, which is decided by the incoming information bits, is time-variant, and $\mathbf{f}$, which determines the reflection pattern, is pre-designed and time-invariant.  For example, when a programmable metasurface is deployed on a building, and its intended receivers are the vehicles, the reflection pattern is designed to cover the road section in front of the building as shown in \figref{SystemModel}.

When the constellation of backscatter communications satisfies the amplitude constraint in \eqref{Const1} and the passive beamforming vector $\mathbf{f}$ satisfies the amplitude constraint in \eqref{ConstF}, the synthesized $\boldsymbol{\Gamma} $ will meet the amplitude constraint and thus can be readily applied to realize directional backscatter communications. Similar to the conventional MIMO with precoding/beamforming \cite{LinearPrecoding, MassLinearPrecoding}, where the design of transmit signals fed to multiple antennas is disentangled into the constellation design and the precoding/beamforming vector design,
the optimization of the reflection coefficient vector $\boldsymbol{\Gamma}$ can be also decomposed into two independent sub-problems of constellation design and passive beamforming design, given that the amplitude constraint is satisfied by the two sub-problems.

To summarize, in this paper, we will carry out (I) the design of the constellation alphabet $\mathcal{S}$ under the amplitude constraint \eqref{Const1} and (II) the design of the reflection pattern (namely, the passive beamforming vector $\mathbf{f}$) under the amplitude constraint \eqref{ConstF} to optimize the performance of directional backscatter communications using programmable metasurface.

\section{APSK Signal Constellation Design Under Amplitude Constraint}
In this section, we resolve the optimization problem of APSK constellation design under amplitude constraint.

\subsection{Amplitude Phase Shift Keying }
Owing to its design flexibility,  APSK has been widely used in practical communications systems, e.g., DVB-S2 \cite{DVB2}, and MIMO precoding design under constant envelope constraint \cite{ Zhang7516597} to provide a power and spectral efficient solution. APSK conveys information through changing both the amplitude and the phase of the carrier signal. APSK with different parameters can be used to represent PSK and QAM. Thus,  APSK can be regarded as a unified modulation scheme. In this paper, we propose to optimize the constellation parameters for APSK in programmable metasurface enabled backscatter communications.

\begin{figure}[tp]
\begin{minipage}[!h]{0.48\linewidth}
\centering
\includegraphics[ width=0.7\textwidth]{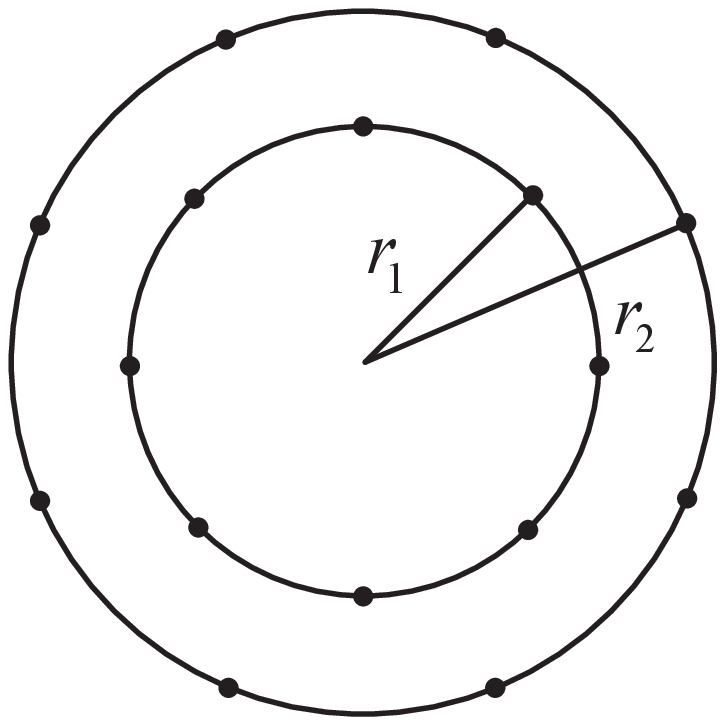}
\subcaption{$N_1=8$, $N_2=8$  }
\label{APSKex1}
\end{minipage}
\begin{minipage}[!h]{0.48\linewidth}
\centering
\hspace{-.63cm}\includegraphics[width=0.7\textwidth]{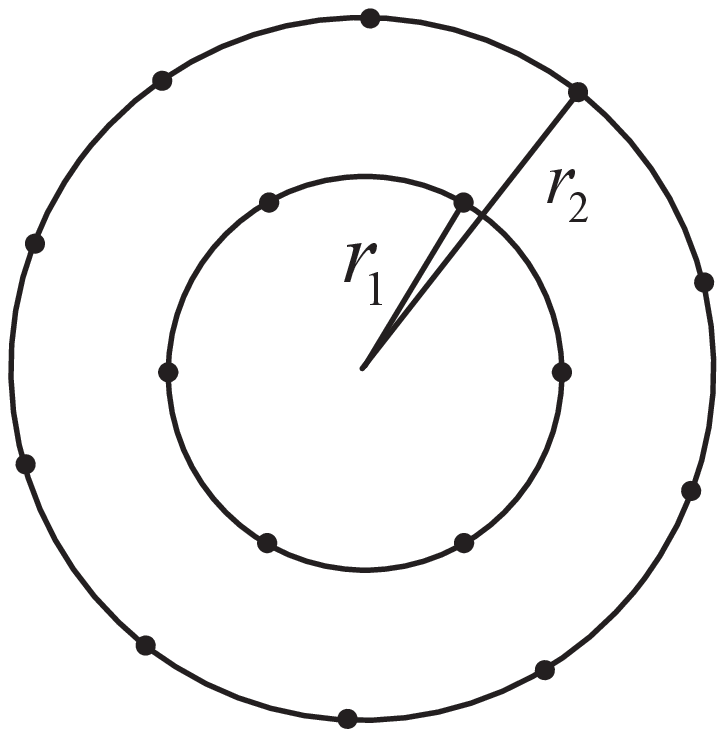}
\subcaption{$N_1=6$, $N_2=10$  \quad \quad}
\label{APSKex2}
\end{minipage}
\caption{Constellation diagrams of the $2$-ring APSK}
\label{APSKex}
\end{figure}

The constellation points of APSK are distributed in two or more concentric rings. For example, $2$-ring 16-APSK constellation with $N_1=8$ points in the outer ring and $N_2=8$ points in the inner ring  is shown in \figref{APSKex1} and $2$-ring 16-APSK constellation with $N_1=10$ points in the outer ring and $N_2=6$ points in the inner ring is shown in \figref{APSKex2}. In a general case, consider an $L$-ring APSK constellation, the elements of which are represented as
\begin{align}
s = r_l e^{j(\frac{2\pi k_l}{N_l} + \omega_l)}
\end{align}
where $l \in \{1,\cdots,L\}$ is the index of the ring,  $L$ is the number of the rings, $k_l \in \{0,\cdots,N_l-1\} $ is the index of the constellation points in the $l$-th ring, $N_l$ is the number of constellation points in the $l$-th ring,  $r_l$ is the radius of the $l$-th ring,  and $\omega_l \in [0, \frac{2\pi}{N_l})$ is the phase of the reference constellation point (i.e., $k_l=0$) in the $l$-th ring. For APSK, the amplitude constraint \eqref{Const1} is rewritten as
\begin{align}
 r_L = \max_{s \in \mathcal{S}} |s|  \leq A \label{Const2}
\end{align}

The performance of signal constellation is dependent on the minimum distance between any pairs of constellation points \cite{wang2016signal}, i.e.,
\begin{align}
d_{min} \triangleq \min_{\substack{s_i, s_j \in \mathcal{S} \\ s_i\neq s_j} }  \left| s_i - s_j \right|  \label{OptFunc}
\end{align}

Thus, combining \eqref{Const2} and \eqref{OptFunc}, the design criterion of  APSK constellation in programmable metasurface enabled backscatter communications is formulated as
\begin{align} \label{Opt0}
\mathrm{P1}: \; \left\{\begin{array} {l}
\max \limits_{\mathcal{S}}   \quad d_{min}  \\
\;  s.t.  \;\; \;\; \; r_L \leq A \\
\;  \;\;  \; \;\; \;\; \; \; \; r_1< r_2<\cdots  < r_L \\
\;  \;\;  \; \;\; \;\; \; \; \; \sum_{l=1}^L N_l=|\mathcal{S}|
\end{array}
\right.
\end{align}

\subsection{The Minimum Euclidean Distance for APSK}
For APSK, the minimum Euclidean distance $d_{min}$ in \eqref{OptFunc} can be further represented as
\begin{align} \label{amin}
d_{min}=\min  \Big\{ &\big \{d_{min\_intra}(l),\; l \in \{1,\cdots,L \} \big\}\cup \notag \\
 & \big\{d_{min\_inter}(l, \hat{l}),\; l \neq \hat{l} \in \{1,\cdots,L\} \big\}  \Big\}
\end{align}
where
\begin{align}
d_{min\_intra}(l) =  \sqrt{2r_l^2- 2r_l^2\cos \frac{2 \pi}{N_l}} \label{dminA}
\end{align}
is the intra-ring minimum Euclidean distance of the $l$-th ring, and
\begin{align}
d_{min\_inter}(l, \hat{l}) = \sqrt{r_l^2 +   r_{\hat{l}}^2 - 2 r_l r_{\hat{l}}\cos  \phi_{l, \hat{l}} } \label{dminB}
\end{align}
is the inter-ring minimum Euclidean distance between the constellation points in the $l$-th ring and  the constellation points  in the $\hat{l}$-th ring, where
\begin{align} \label{angle}
  \phi_{l, \hat{l}}  = \arccos \; \left( \max_{k_l, k_{\hat{l}}} \; \cos \Big(2 \pi \big(\frac{ k_l}{N_l}-  \frac{  k_{\hat{l}} }{N_{\hat{l}}}\big) + \omega_l - \omega_{\hat{l}} \Big) \right) \notag \\
  k_l \in \{0,\cdots,N_l-1\},\; k_{\hat{l}} \in \{0,\cdots,N_{\hat{l}}-1\}
\end{align}
is the minimum phase difference between the constellation points in the two rings. 

\subsection{Decomposition of the Constellation Design}
With \eqref{dminA} and \eqref{dminB}, P1 is rewritten by
\begin{align}
\mathrm{P2}: \; \left\{\begin{array} {l}
\max \limits_{L, \{{r}_l \},  \{\omega_l \}, \{N_{l}\}}    d_{min}  \\
 \;\; \;\; \;\; \;  s.t.  \;\;\;\;   r_L \leq A \\
 \;\; \;\; \;\; \;\;\;  \;\; \;\; \; \;\;    {r}_1< {r}_2  <\cdots <  {r}_L \\
 \;\; \;\; \;\; \;\;\;  \;\; \;\; \; \;\;    d_{min\_intra}(l) \geq d_{min}, \; l \in \{1,\cdots, L \}  \\
  \;\; \;\; \;\; \;\;\;  \;\; \;\; \; \;\; d_{min\_inter}(l, \hat{l}) \geq d_{min},   l \neq \hat{l} \in \{1,\cdots, L \} \\
  \;\; \;\; \;\; \;\;\;  \;\; \;\; \; \;\; \sum_{l=1}^L N_l=|\mathcal{S}| \notag
\end{array}
\right.
\end{align}
where the parameters $L, \{{r}_l \},  \{\omega_l \}, \{N_{l}\} $ are the variables to be optimized. As the ranges of the integer variables  $L, \{N_{l}\} $ are typically small, P2 can be resolved by firstly optimizing $\{r_l \},  \{\omega_l \}$ with a  given set of $L, \{N_{l}\}$, and then exhaustively searching over the feasible region of $L, \{N_{l}\} $. In addition, as the outermost ring corresponds to the best reflectivity, from an energy-greedy perspective, the optimal radius of the outermost ring is apparently $r_L^* = A$. The constraint of inter-ring distance is indeed a solvable linear constraint. However, the constraint of intra-ring distance is non-trivial, which can be written in quadratic form as follows
\begin{align} \label{IntraConst}
\left[
  \begin{array}{cc}
    r_l & r_{\hat{l}} \\
  \end{array}
\right]
 \left[
   \begin{array}{cc}
     1 & -\cos\phi_{l,\hat{l}} \\
     -\cos\phi_{l,\hat{l}} & 1 \\
   \end{array}
 \right]
\left[
  \begin{array}{c}
    r_l \\
    r_{\hat{l}} \\
  \end{array}
\right] \geq d_{min}^2
\end{align}
Obviously, it is  a non-convex  quadratic constraint w.r.t. $[r_l, r_{\hat{l}}]^T$ that renders P2 challenging.

To make P2 tractable, we make the assumption that the minimum Euclidean distance is
the assumptions that  \\
(1) when $N_1 \geq 2$, the minimum Euclidean distance is
\begin{align}
d_{min}= d_{min\_intra}(1) = \sqrt{2 - 2 \cos \frac{2 \pi}{N_1}} r_1 \label{Dmin1}
\end{align}
(2) when $N_1 = 1$ (i.e., there is a constellation point in the center of the ring),  the minimum Euclidean distance is
\begin{align}
d_{min}&= \min\{d_{min\_intra}(2), d_{min\_inter}(1,2)\} \notag \\
&= \min\left\{\sqrt{2 - 2 \cos \frac{2 \pi}{N_1}} r_2, r_2 \right \} \label{Dmin2}
\end{align}

In addition, we can drop the constraint $r_L \leq A$ to obtain a set of intermediate radius parameters ${\hat{r}_l}$, and  then normalize  the intermediate radius parameters as
\begin{align}
r_l = A \frac{\hat{r}_l}{\hat{r}_L}, \;\; \forall l \in \{1, \cdots, L \} \label{Norm}
\end{align}
to meet the constraint. Without loss of generality, we can set $\hat{r}_1 = 1$ (or $\hat{r}_2=1$ when $N_1=1$).

Combining \eqref{Dmin1}, \eqref{Dmin2} with \eqref{Norm}, the objective of maximizing $d_{min}$ is reduced to minimizing $\hat{r}_L$. Thus, with a given set of $L$ and $\{N_l\}$,  the optimization problem can be represented as
\begin{align}
\mathrm{P3}: \; \left\{\begin{array} {l}
\min \limits_{\{\hat{r}_l \},  \{\omega_l \} }  \;\;\;\;\;  \hat{r}_L  \\
 \;\; \;\; \;  s.t.  \;\;\;\;  \hat{r}_1< \hat{r}_2  <\cdots < \hat{r}_L \\
 \;\; \;\; \;\;\;  \;\; \;\; \; \;\;    d_{min\_intra}(l) \geq d_{min}, \; l \in \{1,\cdots, L \}  \\
 \;\; \;\; \;\;\;  \;\; \;\; \; \;\; d_{min\_inter}(l, \hat{l}) \geq d_{min},   l \neq \hat{l} \in \{1,\cdots, L \}  \\
\;\; \;\; \;\;\;  \;\; \;\; \; \;\;   \eqref{Dmin1} \;\; {\rm or}  \;\; \eqref{Dmin2} \notag
\end{array}
\right.
\end{align}

By decomposing P3 into $L-1$ subproblems, i.e.,
\begin{align}
\mathrm{P4}: \; \left\{
\begin{array} {l}
\min \limits_{ \hat{r}_{l+1} ,   \omega_{l+1} }  \;\;\;\;    \hat{r}_{l+1}    \\
\;\; \;\; \;  s.t.  \;\;\;\; \hat{r}_{l+1} > \hat{r}_{l}   \\
\;\; \;\; \;\;\;  \;\; \;\; \; \;\;   {d}_{min\_intra}(l+1)   \geq {d}_{min}  \\
\;\; \;\; \;\;\;  \;\; \;\; \; \;\;  {d}_{min\_inter}(l, l+1) \geq {d}_{min}   \\
 \;\; \;\; \;\;\;  \;\; \;\; \; \;\;  \eqref{Dmin1} \;\; {\rm or}  \;\; \eqref{Dmin2} \notag
\end{array}
\right.
\end{align}
we can solve the problem in a recursive manner starting from the innermost ring ($l=1$) to the outermost ring ($l=L$).  Note that the inter-ring Euclidean distance constraint is relaxed by neglecting ${d}_{min\_inter}(\hat{l}, l+1) \geq {d}_{min}, \; \forall \hat{l} <l $, as the inter-ring Euclidean distance between adjacent rings, i.e., ${d}_{min\_inter}(l, l+1) $,  is usually smaller than the inter-ring Euclidean distance between non-adjacent rings, i.e., ${d}_{min\_inter}(\hat{l}, l+1)$, where $\hat{l} < l$.

\subsection{Solutions to P4}

According to P4,  $\omega_{l+1}$  merely relates to the inter-ring Euclidean distance, while $\hat{r}^*_{l+1} $ is related to both the inter-ring Euclidean distance and the intra-ring Euclidean distance. Thus, we will optimize $\omega_{l+1}$ and $\hat{r}_{l+1} $, successively.

\subsubsection{Optimization of the phase shift $\omega_{l+1}$}
According to \eqref{angle}, $ \phi_{l, l+1 }$ is independent of $r_{l}$  and $r_{l+1}$. Thus, we will firstly maximize $ \phi_{{l, l+1} }$ through optimizing the phase shift $\omega_{l+1}$, which is equivalently to optimize $\Delta\omega_{l, l+1} \triangleq \omega_{l+1} - \omega_{l}$.
\begin{align}
 \max_{\Delta\omega_{l, l+1}}  \phi_{l, l+1}
\end{align}
The analytical expression of the optimal $\Delta\omega_{l, l+1}^*$  and its corresponding $\phi^*_{l, l+1}$ are given in the following proposition.

\begin{proposition} \rm
The optimal phase difference between two adjacent rings is
\begin{align}
\Delta\omega^*_{l, l+1}= \frac{ (1+2\nu) \pi }{{\rm lcm}(N_l, N_{l+1})} , \quad {\rm \nu\; is \; an \; integer}
\end{align}
and the corresponding minimum angle is
\begin{align}
\phi_{l, l+1}({ \Delta\omega_{l, l+1}^* }) = \frac{\pi}{{\rm lcm}(N_l, N_{l+1})}
\end{align}
\end{proposition}

\begin{proof}
See Appendix A
\end{proof}

\subsubsection{Optimization of the radius $\hat{r}_{l+1}$}
According to the constraints  the inter-ring Euclidean distance constraint and the intra-ring Euclidean distance constraint of P4, we derive the range of $r_{l+1} $ as
\begin{subequations} \label{Conds}
\begin{align}
\hat{r}_{l+1}   &\geq \underbrace{\sqrt{\frac{\hat{d}^2_{min}}{2-2\cos\frac{2\pi}{N_{l+1}}} }}_{B_1} \label{Cond1}\\
\hat{r}_{l+1}   &\geq \underbrace{\hat{r}_l \cos \phi^*_{l, l+1}+\sqrt{\hat{r}_l^2 \cos^2(\phi^*_{l, l+1}) - \hat{r}_l^2+ \hat{d}^2_{min}}}_{B_2} \label{Cond2}
\end{align}
\end{subequations}
where the optimal $\phi^*_{l, l+1}$ is obtained by setting $\omega^*_{l+1} = \omega_l + \Delta\omega^*_{l, l+1}$.
Apparently,  the minimum radius is given by
\begin{align} \label{OptRadius}
\hat{r}_{l+1}^* = \max\{B_1, B_2\}
\end{align}

\subsection{The Algorithm for APSK Constellation Construction Under Amplitude Constraint}

With the solutions to P4, parameters of the optimal APSK can be obtained by exhaustively searching over the feasible set of $L$ and $\{N_l\}$. To narrow down the search range, we add the constraint $N_1 \leq N_2 \leq \cdots \leq N_L $ on $\{N_l\}$. To summarize, the procedures of our proposed APSK constellation design are summarized in Algorithm 1.

\begin{algorithm}[ht]

    \caption{Construction of APSK constellation for programmable metasurface enabled backscatter communications}

    \begin{algorithmic}[1]
        \STATEx  \textbf{Input}: Modulation order $M$, modulation order of  PSK in the $1$-st ring $N_1$.
        \STATEx {\textbf{Step 1.}} Find all the possible combinations of $\{N_1, \cdots, N_L \}$ and $L$.

        \STATEx  {\textbf{Step 2.}} For the $k$-th feasible combinations of $\{N_1, \cdots, N_L \}^{(k)}$ and $L^{(k)}$
          \STATEx
             \quad (1) According to Proposition 1 and \eqref{OptRadius}, find the optimal
                 phase shifts $\{\omega_1, \cdots, \omega_L\}^{(k)}$
                 and the optimal intermediate radii  $\{\hat{r}_1, \cdots, \hat{r}_L\}^{(k)}$.
         \STATEx
           \quad (2) Normalize the intermediate radii $\{\hat{r}_1, \cdots, \hat{r}_L\}^{(k)}$, and obtain
            \begin{align}
                r_l^{(k)} = A \frac{\hat{r}_l^{(k)}}{\hat{r}_L^{(k)}}, \;\; \forall l \in \{1, \cdots, L \}
            \end{align}
                  \STATEx
           \quad (3) Record the minimum Euclidean distance $d_{min}^{(k)}$;
         \STATEx
           \quad Go to (1), until all combinations of $\{N_1, \cdots, N_L \}$ and $L$ are exhaustively explored.
            \STATEx  {\textbf{Step 3.}}  Select the combination of $\{N_1, \cdots, N_L \}^{(k)}$ and $L^{(k)}$ which corresponds to the maximum $d_{min}^{(k)}$, and output the corresponding $\{{r}_1, \cdots, {r}_L\}^{(k)}$ and $\{\omega_1, \cdots, \omega_L\}^{(k)}$
    \end{algorithmic}
\end{algorithm}

\begin{remark} {\rm
As the closed-form solution for the subproblem P4 has already been derived, the complexity of the proposed APSK constellation design mainly stems from the exhaustive search for P4 conditioned on $L$ and $\{N_1, N_2, \cdots, N_L\}$ over the feasible set of $L$ and $\{N_1, N_2, \cdots, N_L\}$. Applying the constraint $N_1 \leq N_2... \leq N_L$, the search range of the candidate $\{N_1, N_2, \cdots,N_L\}$ is greatly reduced. In addition, as the constellation is designed off-line and the derived constellation coefficients are pre-stored for impendence selection, the complexity of APSK constellation design is not a critical issue for the real-time backscatter communication.}
\end{remark}

\section{Reflection Pattern Design}
In this section, we firstly analyse the reflection power efficiency of programmable metasurface enabled backscatter communications and then carry out reflection pattern design.

\subsection{Analysis of Reflection Power under Amplitude Constraint}

We assume that the intended angle range is $\mathcal{D}_{\Psi} = [\Psi_x^L, \Psi_x^U) \times [\Psi_y^L, \Psi_y^U)$, and thus the reflected power over $\mathcal{D}_{\Psi}$ is  represented as
\begin{align}
 P_{\mathcal{D}_{\Psi}}  &=  \oiint_{\mathcal{D}_{\Psi}}  \mathbf{f}^H \mathbf{v}(\Psi_x, \Psi_y) \mathbf{v}^H(\Psi_x, \Psi_y)\mathbf{f} \; d \Psi_x d \Psi_y  \notag \\
 &=   \mathbf{f}^H \underbrace{\left( \int_{\Psi_x^L}^{\Psi_x^U} \int_{\Psi_y^L}^{\Psi_y^U}  \mathbf{v}(\Psi_x, \Psi_y) \mathbf{v}^H(\Psi_x, \Psi_y) \; d \Psi_x d \Psi_y \right)}_{\mathbf{V}_{\mathcal{D}_{\Psi}}} \mathbf{f}
\end{align}
According to the mixed-product property of Kronecker product, the term $ \mathbf{V}_{\mathcal{D}_{\Psi}}$ can be represented as
\begin{align}
&\quad \mathbf{V}_{\mathcal{D}_{\Psi}} =  \notag \\
& \underbrace{\left( \int_{\Psi_x^L}^{\Psi_x^U} \mathbf{v}(\Psi_x)\mathbf{v}^H(\Psi_x) \; d \Psi_x  \right)}_{\mathbf{V}_{\Psi_x}} \otimes    \underbrace{\left( \int_{\Psi_y^L}^{\Psi_y^U}  \mathbf{v}(\Psi_y)  \mathbf{v}^H(\Psi_y) \; d \Psi_y \right)}_{\mathbf{V}_{\Psi_y}}
\end{align}
With respect to  $\mathbf{V}_{\Psi_x}$, its $(\ell, \kappa)$-th entry is
\begin{align}
\mathbf{V}_{\Psi_x}(\ell, \kappa) & =  \int^{\Psi_x^U}_{\Psi_x^L}   e^{j (\ell-1) \pi \Psi_x}  e^{ - j (\kappa-1) \pi \Psi_x} d \Psi_x  \notag \\
& = \frac{e^{j (\ell-\kappa) \pi \Psi_x}}{j (\ell-\kappa) \pi }
\bigg|^{\Psi_x^U}_{\Psi_x^L} \notag
\end{align}
which can be further represented as
\begin{align} \label{Vxmatrix}
\mathbf{V}_{\Psi_x}(\ell, \kappa) = \left\{ \begin{array}{c}
                                    \frac{e^{j (\ell-\kappa) \Psi_x^U \pi  }}{j (\ell-\kappa) \pi }  - \frac{e^{j (\ell-\kappa) \Psi_x^L \pi  }}{j (\ell-\kappa) \pi },\;\; \;\;\ell\neq \kappa \\
                                    \Psi_x^U - \Psi_x^L, \;\; \;\;\;\; \;\;\ell=\kappa
                                  \end{array}
 \right.
\end{align}
Similarly, we derive the expression of the $(\ell, \kappa)$-th entry of $\mathbf{V}_{\Psi_y}$ as
\begin{align} \label{Vymatrix}
\mathbf{V}_{\Psi_y}(\ell, \kappa) = \left\{ \begin{array}{c}
                                    \frac{e^{j (\ell-\kappa) \Psi_y^U \pi  }}{j (\ell-\kappa) \pi }  - \frac{e^{j (\ell-\kappa) \Psi_y^L \pi  }}{j (\ell-\kappa) \pi },\;\; \;\;\ell\neq \kappa \\
                                    \Psi_y^U - \Psi_y^L, \;\; \;\;\;\; \;\;\ell=\kappa
                                  \end{array}
 \right.
\end{align}

\begin{proposition} {\rm
The sum power of the reflected signals in all directions ( i.e., $\mathcal{D}_{\Psi} =  [-1, 1) \times [-1, 1)$) is proportional to the squared $\ell_2$ norm of $\mathbf{f}$, i.e., $  E_{[-1, 1) \times [-1, 1)} \varpropto \mathbf{f}^H \mathbf{f}$. }
\end{proposition}
\begin{proof}
See Appendix B
\end{proof}

\begin{remark}
{\rm The maximum power of the radiated signal in all directions for conventional beamforming can be achieved by simply normalizing the beamforming vector as $\mathbf{f}^H \mathbf{f} = N$ under sum power constraint, while passive beamforming under amplitude constraint has to be of constant modulus, i.e., $|\mathbf{f}(i)| = 1, \; \forall n \in \{1, \cdots, N \}$, to achieve the maximum reflection power in all directions. }
\end{remark}

\subsection{Compatibility of Conventional Beam Pattern Designs Under Sum Power Constraint With Programmable Metasurface Enabled Passive Beamforming}
We firstly review conventional beamforming techniques under sum power constraint, including fully digital beamforming, hybrid beamforming with multiple radio frequency (RF) chains, and hybrid beamforming with a single RF chain, and then discuss their compatibility with programmable metasurface enabled passive beamforming.

\vspace{0.08cm}
\emph{1) Fully digital beamforming }
\vspace{0.08cm}

\textbf{Review:} Beam pattern design for fully digital beamforming has been well investigated in array signal processing \cite{long2019window,cheng2021analytical}. Similar to the design of finite impulse response (FIR) filters \cite{FIRls}\cite{ConstLS},  window-based method \cite{long2019window} and least-square (LS) (or constrained least-square (CLS)) method \cite{cheng2021analytical, alkhateeb2014channel} can be readily applied to beam pattern designs. To accommodate the hardware structure, the length of the response needs to be equal to the number of array elements.

\textbf{Compatibility:} Beamforming vectors, including fully digital case, can meet the amplitude constraint of programmable metasurface enabled passive beamforming through normalizing $\mathbf{f}$ as follows
\begin{align}
\mathbf{f} \leftarrow \frac{\mathbf{f}}{\max_n |\mathbf{f}(n)|} \label{Normalization}
\end{align}
The normalized beamforming vector satisfies $\mathbf{f}^H \mathbf{f} = \frac{N}{\max_n |\mathbf{f}(n)|^2}$. According to Proposition 2, the sum power of the reflected signals is restricted by $\max_n |\mathbf{f}(n)|$, and thus a very large $\max_n |\mathbf{f}(n)|$ will result in a power-inefficient reflection pattern.

\vspace{0.08cm}
\emph{2) Hybrid beamforming with multiple RF chains}
\vspace{0.08cm}

\textbf{Review:} Hybrid beamforming vector needs to be compatible with the hybrid digital and analog hardware structure. The key idea of beam pattern design for hybrid beamforming with multiple RF chains is to regenerate the reference beamforming vector, which is derived in fully digital case, under the hardware constraint. In \cite{alkhateeb2014channel}, the reference beam pattern is firstly obtained through LS method and then approximated using orthogonal matching pursuit (OMP) algorithm. 
In \cite{wang2020optimal}, the reference beam pattern is firstly obtained through semidefinite relaxation (SDR) technique and then perfectly regenerated through vector decomposition.

\textbf{Compatibility:} As hybrid beamforming vector is an approximation of the desired digital beamforming vector (i.e., reference beamforming vector), its extension to reflection pattern design is the same as fully digital case.

\vspace{0.08cm}
\emph{3) Analog beamforming with a single RF chain}
\vspace{0.08cm}

\textbf{Review:} Analog beamforming design with a single RF chain is performed under constant modulus constraint, which is brought by the analog phase shifter network. In \cite{xiao2016hierarchical,zhu20193}, subarray based beamforming designs are carried out for single-RF-chain array antenna. Through set partition, the element antennas are grouped as subarrays, and then subarrays with different directions are combined to synthesize beam patterns with different beamwidths. However,  beamwidth of the subarray based method is confined to a few discrete values, which inevitably introduces a mismatch between the desired beamwidth and the actual beamwidth. In addition,  half of the array elements have to be deactivated to attain some beamwidths using subarray based method.

\textbf{Compatibility:} The element of analog beamforming vector is either of constant modulus or zero-valued (deactivated). Analog beamforming vector can be readily applied to passive beamforming. The subarray based method \cite{xiao2016hierarchical,zhu20193} generates two types of analog beamforming vectors. For Type 1, where all the analog phase shifters are activated, the beamforming vector satisfies $\mathbf{f}^H \mathbf{f} = N$; for Type 2, where half of the analog phase shifters are deactivated, the beamforming vector satisfies $\mathbf{f}^H \mathbf{f} = \frac{N}{2}$. Type 1 attains the maximum achievable reflection power, while Type 2 attains only half of the maximum achievable reflection power.

\begin{remark}{\rm
Although the aforementioned off-the-shelf methods can be readily applied to passive beamforming after the normalization operation of \eqref{Normalization}, the cost of the inefficient reflection power is prohibitively expensive for backscatter communications.
}
\end{remark}

\subsection{Reflection Pattern Design Under Amplitude Constraint}
To fully exploit the reflectivity of programmable metasurface and achieve the maximum reflection power, we tighten the amplitude constraint $ |\mathbf{f}(i)| \leq 1, \;\forall n \in \{1, .., N\}$ and incorporate the constant modulus constraint $ |\mathbf{f}(i)| = 1, \;\forall n \in \{1, .., N\} $ to reflection pattern design.

In addition, the reflection pattern design also aims to achieve the following two objectives.

\underline{Objective 1:}  Maximize the sum power of reflected signal over the intended angle range
\begin{align}
&\max \limits_{\mathbf{f}} \; \mathbf{f}^H \mathbf{V}_{\mathcal{D}_{\Psi} }  \mathbf{f}  \notag\\
&\;\;s.t.  \;\; |\mathbf{f}(i)| = 1, \;\forall n \in \{1, .., N\} \notag
\end{align}

\underline{Objective 2:} Maximize the minimum power of reflected signal within the intended angle range

\begin{align}
&\max \limits_{\mathbf{f}} \; \min \limits_{ (\Psi_x, \Psi_y) \in \mathcal{G} } \quad \mathbf{f}^H \mathbf{v}(\Psi_x, \Psi_y) \mathbf{v}^H(\Psi_x, \Psi_y) \mathbf{f}  \notag \\
&\qquad s.t.\;  \qquad \qquad   |\mathbf{f}(i)| =  1, \;\forall n \in \{1, .., N\} \notag
\end{align}
where
\begin{align}
&\mathcal{G} \triangleq \left\{(\Psi_x, \Psi_y) \Big| \; \Psi_x = \Psi_x^L +    \frac{2}{N_x} n_x, \Psi_y = \Psi_y^L + \frac{2}{N_y} n_y, \right.\notag \\
&\left. n_x = 0, \cdots, \big \lfloor \frac{\Psi_x^U - \Psi_x^L}{2/N_x} \big \rfloor -1, n_y = 0, \cdots, \big\lfloor \frac{\Psi_y^U - \Psi_y^L}{2/N_y} \big\rfloor -1 \right\}  \notag
\end{align}
is the discrete grid of $(\Psi_x, \Psi_y)$ over the intended range $[\Psi_x^L, \Psi_x^U) \times [\Psi_y^L, \Psi_y^U)$ with the grid size $(\frac{2}{N_x}, \frac{2}{N_y})$.

Note that Objective 1  and Objective 2 correspond to the design criterion 1 and design criterion 2 in Section II. C, respectively.  To resolve the above multi-objective optimization problem, we apply the weighted sum method \cite{marler2010weighted}. Specifically, we introduce a hyper-parameter $\alpha$ to combine  Objective 1 and Objective 2 and formulate the new research problem as follows
\begin{align}
\mathrm{P5}: \; \left\{
\begin{array} {l}
\max\limits_{\mathbf{f}} \; \min\limits_{ (\Psi_x, \Psi_y) \in \mathcal{G} } \quad \mathbf{f}^H  \mathbf{M}_{\Psi_x, \Psi_y} \mathbf{f}  \\
 \;\; \;\; \;  s.t. \qquad \;\;\;\; |\mathbf{f}(i)| =  1, \;\forall n \in \{1, .., N\}
\end{array}
\right.
\end{align}
where  $\mathbf{M}_{\Psi_x, \Psi_y} \triangleq \mathbf{v}(\Psi_x, \Psi_y) \mathbf{v}^H(\Psi_x, \Psi_y) + \alpha \mathbf{V}_{\mathcal{D}_{\Psi} }$.

\vspace{0.1cm}
For the angle range $\mathcal{D}_{\Psi} = [\Psi_x^L, \Psi_x^U) \times [\Psi_y^L, \Psi_y^U)$,  P5 can be broken down into the subproblems w.r.t. $\Psi_x$ and $\Psi_y$. W.r.t. $\Psi_x$, the design problem is given as
\begin{align}
\mathrm{P6}: \; \left\{
\begin{array} {l}
\max\limits_{\mathbf{f}_x} \; \min\limits_{  \Psi_x \in \mathcal{G}_x } \quad \mathbf{f}_x^H  \mathbf{M}_{\Psi_x} \mathbf{f}_x  \\
 \;\; \;\; \;  s.t. \qquad  \;\;\;\; |\mathbf{f}_x (i)| =  1, \;\forall n_x \in \{1, .., N_x\}
\end{array}
\right.
\end{align}
W.r.t $\Psi_y$, the component vector $\mathbf{f}_y$ can be obtained in the similar way. Then, the beamforming vector can be derived as $\mathbf{f} = \mathbf{f}_x \otimes \mathbf{f}_y$.

\subsection{Solution to P6}
We firstly focus on maximizing a specific term $ \mathbf{f}_x^H  \mathbf{M}_{\Psi_x} \mathbf{f}_x $ and then extend the method to the max-min problem \cite{wang2019spatial}.

\subsubsection{Constant-Modulus Power Iteration Method (CMPIM) to Maximize $\mathbf{f}_x^H  \mathbf{M}_{\Psi_x} \mathbf{f}_x $}
The term $ \mathbf{f}_x^H  \mathbf{M}_{\Psi_x} \mathbf{f}_x $ can be maximized through the following iterative process
\begin{subequations} \label{Iteration}
\begin{align}
    \mathbf{f}_{x, temp}^{(i)} & = \mathbf{f}_x^{(i)} +  \delta    \mathbf{M}_{\Psi_x}  \mathbf{f}_x^{(i)} \label{Step1a}\\
    {\mathbf{f}_x^{(i+1)}} & = \mathbf{D} \mathbf{f}_{x, temp}^{(i)} \label{Step1b}
\end{align}
\end{subequations}
where $\delta  \in (0, \infty] $ is the step size, and
\begin{align}
\mathbf{D}=  \diag \left\{ \frac{1}{|\mathbf{f}_{x, temp}^{(i)}(1)|},\;\cdots, \;   \frac{1}{|\mathbf{f}_{x, temp}^{(i)}(N_x)|} \right\}  \notag
\end{align}

\begin{proposition} {\rm
The term  $ \mathbf{f}_x^H  \mathbf{M}_{\Psi_x} \mathbf{f}_x $ in constant-modulus power iteration method (namely, \eqref{Iteration}) is monotonically increasing, i.e.,
\begin{align}
    {\mathbf{f}_x^{(i+1)}}^H \mathbf{M}_{\Psi_x} \mathbf{f}_x^{(i+1)}  \geq {\mathbf{f}_x^{(i)}}^H \mathbf{M}_{\Psi_x}  {\mathbf{f}_x^{(i)}}, \;\; \forall \delta \in (0, \infty]
\end{align}
and the equality holds if and only if $\mathbf{f}_x^{(i+1)} = \mathbf{f}_x^{(i)}$. }
\end{proposition}

\begin{proof}
See Appendix C.
\end{proof}

\begin{lemma} {\rm
The constant-modulus power iteration method of \eqref{Iteration} is convergent. }
\end{lemma}
\begin{proof}
According to monotone convergence theorem, a monotone and bounded sequence is convergent. In Proposition 3, ${\mathbf{f}_x^{(i)}}^H \mathbf{M}_{\Psi_x}  {\mathbf{f}_x^{(i)}}$ is proven to be monotonically increasing.
It is also easy to find that $\mathbf{f}_x^H  \mathbf{M}_{\Psi_x} \mathbf{f}_x$ is upper bounded by $N_x \lambda_{M,1}$, where $\lambda_{M,1}$ is the largest eigenvalue of $\mathbf{M}_{\Psi_x}$.  Then, we can conclude that the  constant-modulus power iteration method is convergent.
\end{proof}

\begin{remark}{\rm
The ascent rate is controlled by the step size $\delta$. When $\delta  \rightarrow \infty$,  constant-modulus power iteration method is merely different from the classical power iteration method \cite{watkins2004fundamentals} in vector normalization. Constant-modulus power iteration method applies amplitude normalization (namely, \eqref{Step1b}), while the classical power iteration method  applies power normalization.  
}
\end{remark}

\subsubsection{Algorithm to Resolve P6}
On the basis of constant-modulus power iteration method, P6 can be resolved through iteratively updating $\mathbf{f}_x$ to increase the value of the minimum term $\min\limits_{\Psi_x \in \mathcal{G}_x } \quad \mathbf{f}_x^H  \mathbf{M}_{\Psi_x} \mathbf{f}_x$ in each iteration (Algorithm 2).  However, the iteration might very likely to cause a sharp decrease of other terms when the step size $\delta$ is very large. To this end, a relatively small step size is desirable to guarantee convergence.

\begin{algorithm}[ht]

    \caption{Algorithm to solve P6}

    \begin{algorithmic}[1]
        \STATEx  \textbf{Initialization}: Set the step size $\delta$ and the error tolerance $\epsilon$. Randomly initialize $\mathbf{f}_x^{(i)},  (i = 0)$ under the constant modulus constraint.
        \STATEx {\textbf{Step 1.}} Find the angle  $ \Psi_x^* $ that corresponds to the minimum level of power radiation, i.e.,
            \begin{align}
                 \Psi_x^* = \argmin_{ \Psi_x \in \mathcal{G}_x}     {\mathbf{f}_x^{(i)}}^H \mathbf{M}_{\Psi_x} \mathbf{f}_x^{(i)}
            \end{align}
        \STATEx  {\textbf{Step 2.}} Apply constant-modulus power iteration method  to update $\mathbf{f}_x$, i.e.,
        \begin{subequations}
        \begin{align}
        \mathbf{f}_{x, temp}^{(i)} & = \mathbf{f}_x^{(i)} +  \delta    \mathbf{M}_{\Psi_x}  \mathbf{f}_x^{(i)} \notag\\
        {\mathbf{f}_x^{(i+1)}} & = \mathbf{D} \mathbf{f}_{x, temp}^{(i)}  \notag
        \end{align}
        \end{subequations}
        where $ \mathbf{D}=  \diag \left\{ \frac{1}{|\mathbf{f}_{x, temp}^{(i)}(1)|},\;\cdots, \;   \frac{1}{|\mathbf{f}_{x, temp}^{(i)}(N_x)|} \right\}$.
        \STATEx
        Go to Step 1 until $\| \mathbf{f}_x^{(i)} - \mathbf{f}_x^{(i+1)} \| \leq \epsilon $, and set $i \leftarrow i+1$.
    \end{algorithmic}
\end{algorithm}

\begin{remark}{\rm
Due to the non-convexity of P6, the result of Algorithm 2 might be a local optimum. To reduce the chance of being trapped in local optimum, we need to run Algorithm 2 with random initial points for multiple times and choose the best result.
}
\end{remark}

\begin{remark}{\rm
In the LoS-dominant channel, the directional backscatter communications achieve significantly better performance than the non-directional backscatter communications. However, in the complex scattering environments with substantial reflection, reverberation,  multi-path, etc., the performance gain brought by directional beamforming degrades.
}
\end{remark}

\section{Numerical Results }

In this section, we present some numerical results to verify the effectiveness of our proposed APSK constellation design in programmable metasurface enabled backscatter communications.

\subsection{Numerical Study of Constellation Design }
The minimum Euclidean distance between different constellation points is a key performance indicator of the constellation, which determines its symbol error rate (SER), bit error rate (BER) and mutual information \cite{xiao2011globally}. To this end, we make comparisons of the minimum Euclidean distance between the conventional QAM/PSK constellations and our optimized APSK constellation. We firstly list the parameters for the optimized APSK as follows and plot the  constellations  in \figref{APSK}.
\begin{itemize}
\item
When the modulation order is $M = 8$, the  parameters for the optimal APSK are $L=2$, $N_1 = 1, N_2 = 7$, $r_1 = 0, r_2 = 1$, and $\omega_1 = 0, \omega_2 = 0.4488$;
\item
When the modulation order is $M = 16$, the parameters for the optimal APSK are $L=2$, $N_1 = 5, N_2 = 11$,  $r_1 = 0.4603, r_2 = 1$,  $\omega_1 = 0, \omega_2 = 0.0571$;
\item
When the modulation order is $M = 32$, the parameters for the optimal APSK are $L=3$, $N_1 = 5, N_2 = 10, N_3 = 17$, $ r_1 = 0.3068, r_2 = 0.6397, r_3 = 1$, $\omega_1 = 0, \omega_2 = 0.3142, \omega_3= 0.3326$;
\item
When the modulation order is $M = 64$, the parameters for the optimal APSK are $L=5$, $N_1 = 1, N_2 = 6, N_3 = 13, N_4 = 19, N_5 = 25$, $r_1 = 0, r_2 =  0.2446, r_3 = 0.5110, r_4 =  0.7555, r_5 = 1$, $\omega_1 = 0, \omega_2 = 0.5236, \omega_3= 0.5639, \omega_4= 0.5766, \omega_5= 0.5832$;
\end{itemize}
We present the comparisons of minimum Euclidean distance $d_{\min}$ between APSK, PSK, and QAM under amplitude constraint in Table \ref{AmCons}, from which we can see that the optimized APSK is superior to both PSK and QAM when $M = 8, 16, 32, 64$. QAM is a type of dense constellation that follows the structure of $\mathbf{Z}^2$ lattice \cite{wang2016signal}, and in power-constrained case, QAM usually achieves satisfying $d_{min}$ performance. For comparative purposes, we also list the $d_{\min}$ of the same QAM, PSK and APSK constellations under power constraint in Table \ref{PowerCons}. As can be seen that,  QAM outperforms APSK when $M=16, 32$ under power constraint. To conclude, Table \ref{AmCons} and  Table \ref{PowerCons} jointly indicate that our proposed design for APSK constellation is an efficient scheme under amplitude constraint.

\begin{table}[t] \centering
\centering
\caption{The comparison of $d_{min}$ for PSK, QAM and the optimal APSK under amplitude constraint}\label{AmCons}
\begin{tabular}{|l|*{7}{c|}}\hline
 \makebox[6.5em] {Modulation Order }
&\makebox[3.5em]{PSK} &\makebox[3.5em]{QAM}  &\makebox[4.5em]{APSK} \\\hline\hline
\qquad $M=8$ \;  & 0.7654  &  0.6325  &  0.8678  \\\hline
\qquad $M = 16$ \;  & 0.3902  &  0.4714  &  0.5411  \\\hline
\qquad  $M = 32$ \;  &  0.1960    & 0.3430  &  0.3606  \\\hline
\qquad  $M = 64$ \;  &  0.0981  &  0.2020  & 0.2446  \\\hline
\end{tabular}
\vspace{0.2cm}
\centering
\caption{The comparison of $d_{min}$ for PSK, QAM and the optimal APSK under power constraint}\label{PowerCons}
\begin{tabular}{|l|*{7}{c|}}\hline
 \makebox[6.5em] {Modulation Order }
&\makebox[3.5em]{PSK} &\makebox[3.5em]{QAM}  &\makebox[4.5em]{APSK} \\\hline\hline
\qquad $M=8$ \;  & 0.7654  &  0.8165  &  0.9277  \\\hline
\qquad $M = 16$ \;  & 0.3902  &  0.6325  &  0.6233  \\\hline
\qquad  $M = 32$ \;  &  0.1960    & 0.4472  &  0.4393  \\\hline
\qquad  $M = 64$ \;  &  0.0981  &  0.3086  & 0.3109  \\\hline
\end{tabular}
\end{table}

\begin{figure}

\begin{minipage}[!h]{0.48\linewidth}
\centering
\includegraphics[ width=1.1\textwidth]{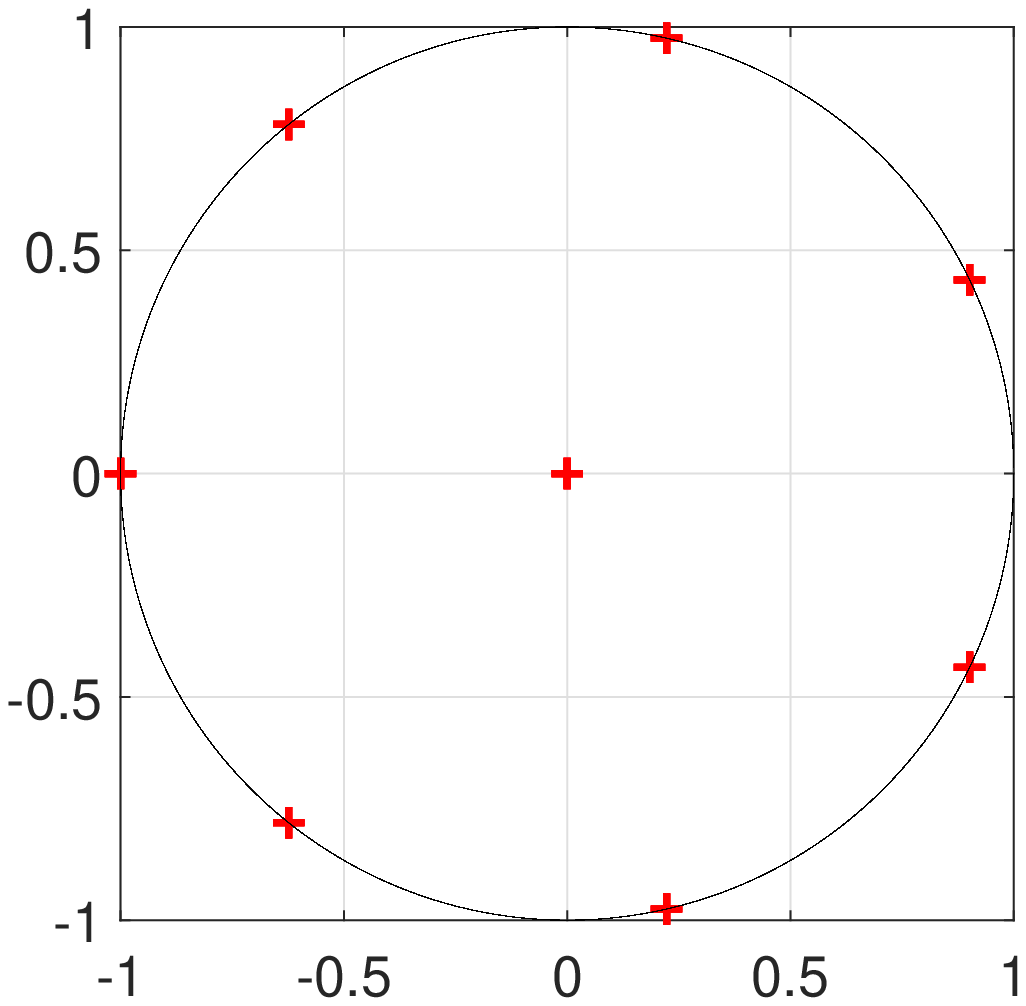}
\subcaption{$M=8$ }
\label{APSK1}
\end{minipage}
\begin{minipage}[!h]{0.48\linewidth}
\centering
\hspace{-.63cm}\includegraphics[width=1.1\textwidth]{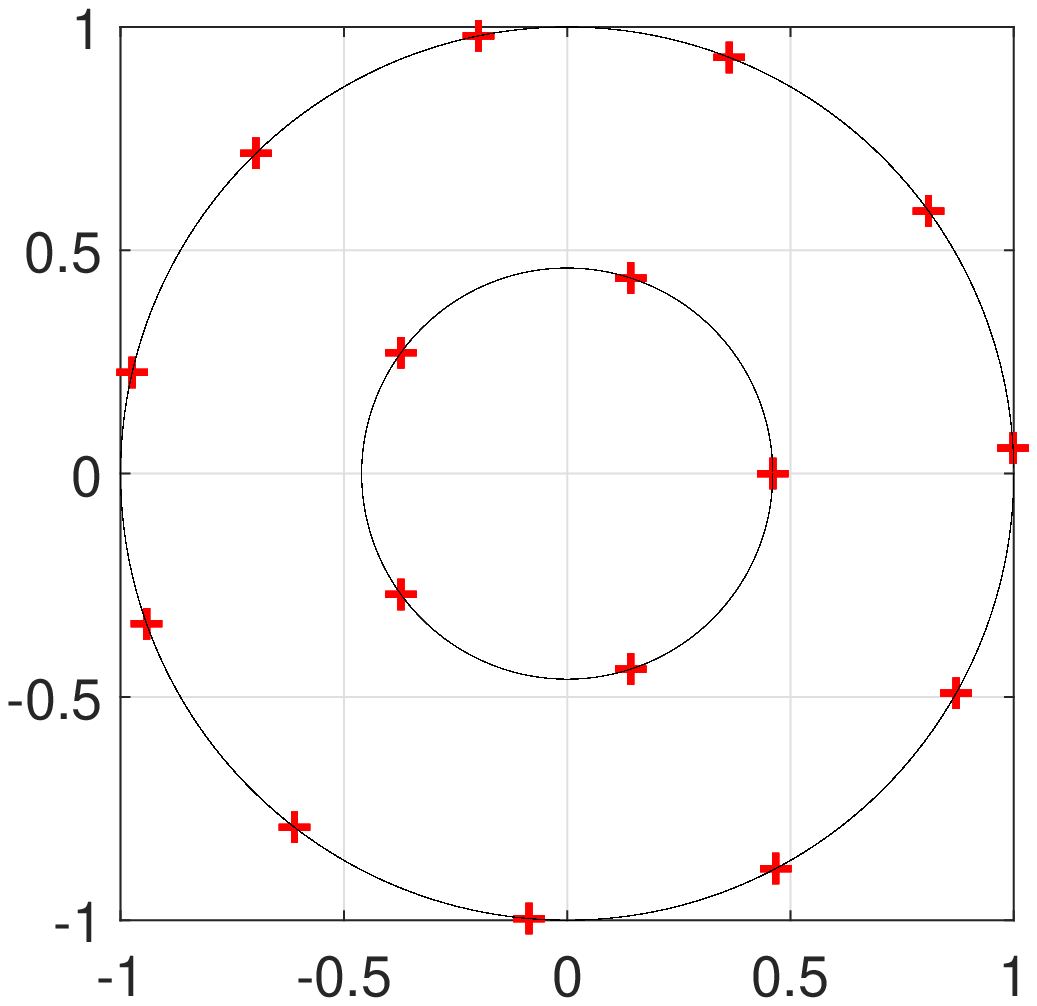}
\subcaption{$M=16$ }
\label{APSK2}
\end{minipage}
\begin{minipage}[!h]{0.48\linewidth}
\centering
\includegraphics[ width=1.1\textwidth]{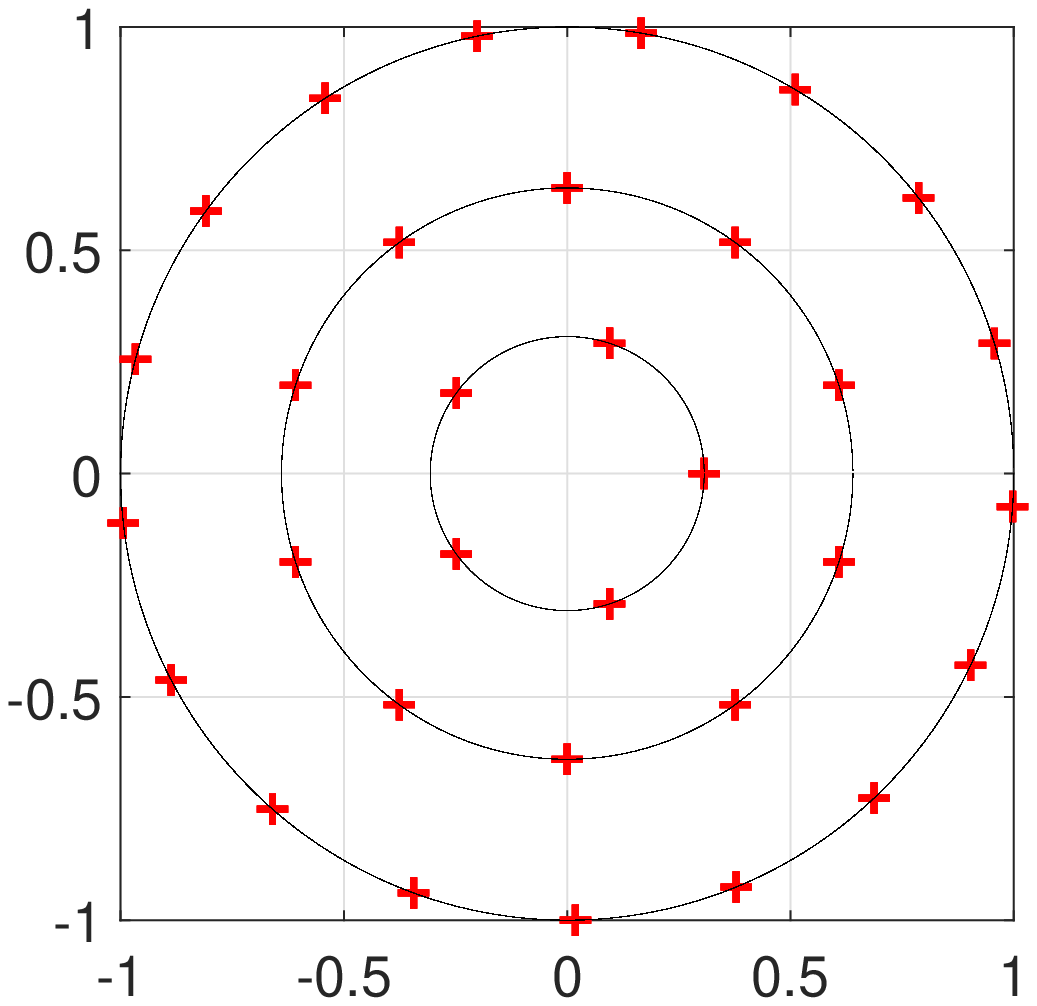}
\subcaption{$M=32$ }
\label{APSK1}
\end{minipage}
\hspace{-.3cm}\begin{minipage}[!h]{0.48\linewidth}
\centering
\includegraphics[width=1.1\textwidth]{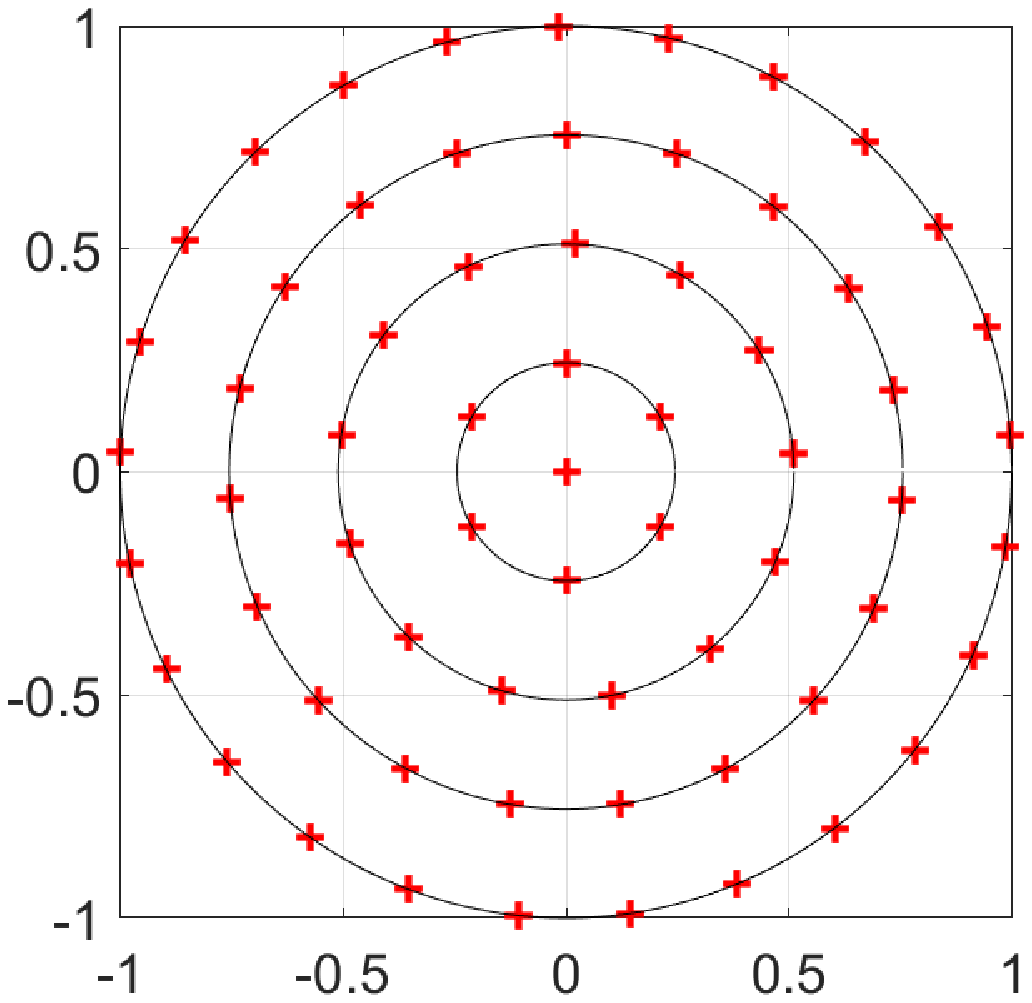}
\subcaption{$M=64$ }
\label{APSK2}
\end{minipage}
\caption{Constellation diagrams of the optimized APSK}
\label{APSK}
\end{figure}

In \figref{BER}, we study the BER of the optimized APSK in additive white Gaussian noise (AWGN) channel, where the $x$-axis represents $\frac{E_b}{N_0}$, i.e., energy per bit to noise power spectral density ratio, and the $y$-axis represents BER. From the figure, we can see that the optimized APSK achieves better BER performance than both QAM and PSK when the modulation order is $M=8, 16, 32, 64$. Specifically, the performance enhancement is approximately $1$dB when $M=8, 16$, $0.5$dB when $M=32$, and $1.5$dB when $M=64$. It indicates that when programmable metasurface enabled backscatter communications adopt high-order modulations, our proposed design consumes $26\%$, $26\%$, $12\%$, and $40\%$ less energy from the incident power source.

\subsection{Numerical Study of Reflection Pattern Design}

\begin{figure}[tp]{
\begin{center}{\includegraphics[ height=7cm]{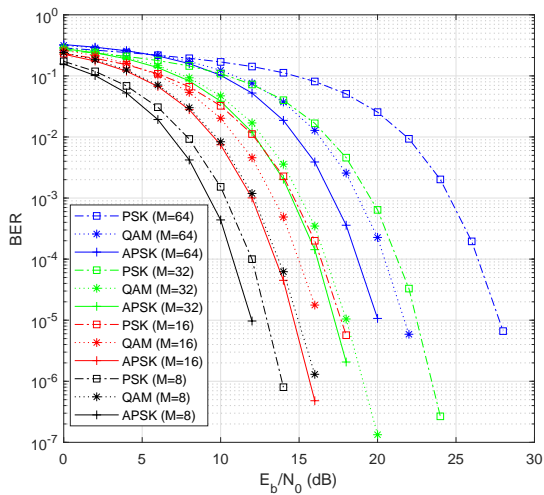}}
\caption{BER comparisons between PSK, QAM and APSK for programmable metasurface enabled backscatter communications}\label{BER}
\end{center}}
\end{figure}

\begin{figure}[tp]{
\begin{center}{\includegraphics[ height=5.5cm]{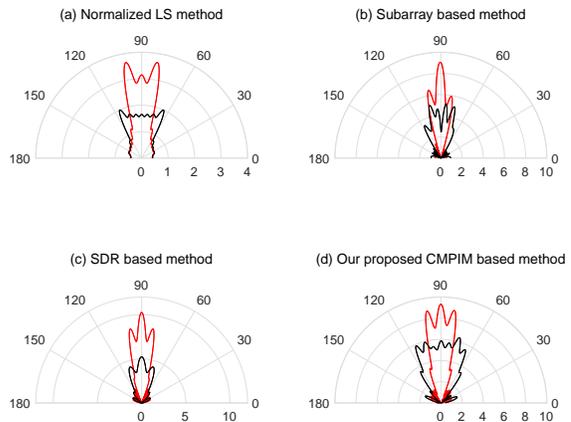}}
\vspace{0.2cm}
\caption{1-D Beam patterns of the component passive beamforming vectors (Black
curve corresponds to the component passive beamforming vector $\mathbf{f}_x$
with the angle range $[-0.5, 0.5)$, and red curve corresponds to the component passive beamforming vector $\mathbf{f}_y$ with the angle range $[-0.25, 0.25)$) } \label{BeamPattern}
\end{center}}
\end{figure}

{\em  Benchmark Schemes and Simulation Parameters:} To study the performance of our proposed CMPIM based method, we make a comparison with three benchmark schemes \cite{alkhateeb2014channel, wang2020optimal, xiao2016hierarchical}.
\begin{itemize}
\item In Benchmark 1 (termed as normalized LS method), the passive beamforming vector is derived through LS estimation of the ideal rectangular beam pattern, which is similar to \cite{alkhateeb2014channel}, and then normalized according to \eqref{Normalization}.
\item In Benchmark 2 (termed as SDR based method), the passive beamforming vector is derived by solving a max-min optimization problem using SDR technique, and different from the traditional sum power constraint case  \cite{wang2020optimal, luo2010semidefinite}, the generated Gaussian randomizations is normalized according to \eqref{Normalization}.
\item Benchmark 3 (termed as subarray based method) follows the design in \cite{xiao2016hierarchical}.
\end{itemize}
As the beamwidth of the subarray based method is confined to a few discrete values, we set the intended angle range in the numerical study as $\mathcal{D}_{\Psi} = [-0.5, 0.5) \times [-0.25, 0.25)$ (which corresponds to AoD range $\psi_x \in (60^{\degree}, 120^{\degree}], \psi_y \in (75.52^{\degree}, 104.48^{\degree}]$) for fairness in comparisons.  In addition, we set the number of metasurface units as $N = N_x \times N_y = 16 \times 16$ and
the inter-element spacing as half the wavelength.

\begin{figure}[tp]{
\begin{center}{\includegraphics[ height=6.5cm]{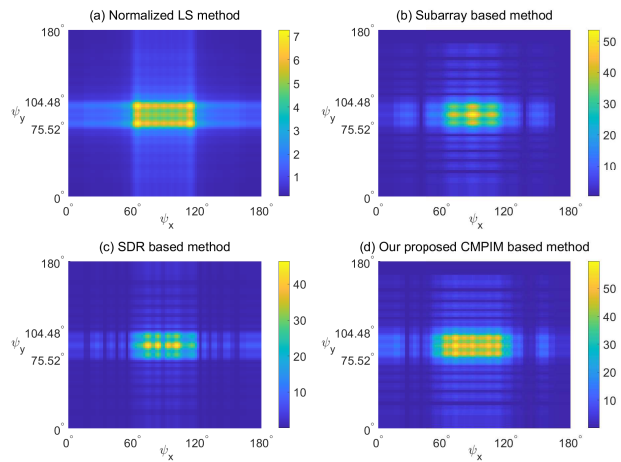}}
\caption{2-D beam patterns of the passive beamforming vector $\mathbf{f}_{x,y}$}\label{2Dpattern}
\end{center}}
\end{figure}

\begin{table}[htp] \centering
\noindent
\caption{Power Ratio and Ripple Factor} \setlength{\belowcaptionskip}{0cm} \label{T1}
 \begin{tabular} {|p{3.45cm}<{\centering}|p{1.6cm}<{\centering}|p{1.6cm}<{\centering}|}
\hline\hline
\backslashbox{Beam Pattern}{Performance}   & Ripple Factor    & Power Ratio \\ \hline
Normalized LS method &  0.1410    &   1.34\%    \\\hline
SDR based method \cite{wang2020optimal} &  0.3299  &  27.50\%    \\\hline
Subarray based method \cite{xiao2016hierarchical}  &  0.3184 & 42.78\%  \\\hline
CMPIM based method &  0.2259  &   81.45\%  \\\hline
\end{tabular}
\end{table}

In \figref{BeamPattern}, the 1-D beam patterns of the component passive beamforming vectors $\mathbf{f}_x$ and $\mathbf{f}_y$ generated by different methods are presented.  In \figref{2Dpattern}, the 2-D beam patterns of the passive beamforming vector $\mathbf{f}$ generated by different methods are presented. In \figref{BeamPattern} amplitude is proportional to the radial distance from the center point and in \figref{2Dpattern} amplitude is represented by color scale (For illustrative purposes, sub-figures are displayed with different scales).  Based on the figures, beam patterns generated by our method are more power efficient and more flat in passband. To quantitatively validate the our observations,  we derive the ripple factor and power ratio of the 2-D beam patterns in Table \ref{T1}, where ripple factor is defined in \eqref{Subeq4}, and power factor is defined as the ratio of $E_{ [-0.5, 0.5) \times [-0.25, 0.25)}$ to the maximum achievable reflection power in all directions, i.e., $E^{\mathbf{f}^H \mathbf{f} = N}_{[-1, 1) \times [-1, 1)} = 4N$. We can see that the normalized LS method achieves the smallest ripple factor, followed by our proposed CMPIM based method, while SDR based method and subarray based method are the worst in ripple factor and experience drastic fluctuation in the passband. It is noteworthy that, unlike the beam pattern design under sum power constraint in \cite{wang2020optimal}, ripple factor of SDR based method deteriorates significantly due to the amplitude constraint. As for power ratio, our proposed CMPIM based method is the most efficient in passive beamforming, which achieves $81.45\%$ of the maximum reflection power within intended angle range, subarray based method achieves $42.78 \%$ of the maximum reflection power, SDR based method achieves $27.50 \%$ of the maximum reflection power, and normalized LS method achieves merely $1.34 \%$ of the maximum reflection power. Although most of the reflection power falls into the passband for the three benchmark designs, their power ratios are still unsatisfying. It is because their generated passive beamforming vectors satisfy $\mathbf{f}^H\mathbf{f} << N$. Specifically, the power inefficiency of SDR based method and normalized LS method are caused by the amplitude normalization operation \eqref{Normalization}, and subarray based method is due to the deactivation of half of the metasurface units for the component passive beamforming vector $\mathbf{f}_y$ whose passband $[-0.5, 0.5)$. In a word, unlike traditional MIMO beamforming designs, the programmable metasurface enabled passive beamforming has to meet $|\mathbf{f}(i)| = 1, \;\forall n \in \{1, .., N\}$ to maximize the power efficiency of signal reflection.

\begin{figure}[tp]{
\begin{center}{\includegraphics[ height=7cm]{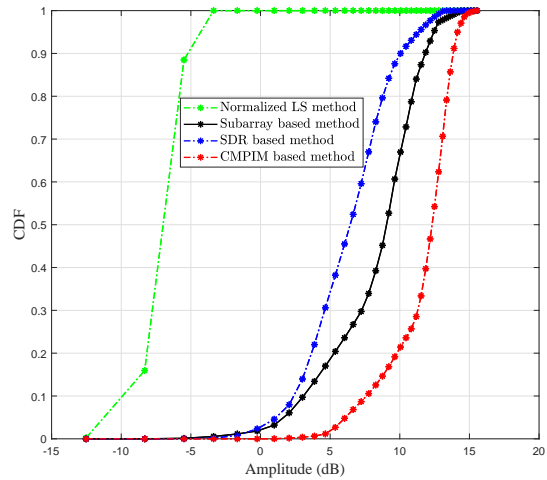}}
\vspace{0.2cm}
\caption{CDF of beam amplitude within the intended angel range} \label{CDF}
\end{center}}
\end{figure}

In \figref{CDF},  we numerically analyze the signal coverage of different reflection patterns within the intended angle range $[-0.5, 0.5) \times [-0.25, 0.25)$. The X-axis represents beam amplitude in dB, i.e., $20\log_{10} \frac{Amp}{10}$, and the Y-axis represents the cumulative distribution functions (CDF) of beam amplitude. The CDF curves are derived by sampling over the independent uniform distributions $\Psi_x \sim U(-0.5, 0.5), \Psi_y \sim U(-0.25, 0.25)$. From the figure, we can see that the beam amplitude  is primarily within the range  [-12.5dB, -4dB]  for normalized LS method, [-5dB, 13dB] for subarray based method, [-5dB, 14dB] for subarray based method, and [3dB, 15dB] for the proposed CMPIM based method.  It is noteworthy that the $8.5$dB amplitude span, $18$dB amplitude span, $19$dB amplitude span, and $12$dB amplitude span of the four methods are in accordance with their ripple factors in Table. \ref{T1}. The narrower amplitude span means the better performance stability.  Besides, we can also find that $80\%$ of the angles $(\Psi_x, \Psi_y)$ achieve greater than $10$dB amplitude in our proposed CMPIM based method, while in the best benchmark scheme, i.e., subarray based method, only $30\%$ of the angles achieve greater than $10$dB amplitude.

\subsection{Numerical Study of Directional Backscatter Communications}

\begin{figure}[tp]{
\begin{center}{\includegraphics[ height=7cm]{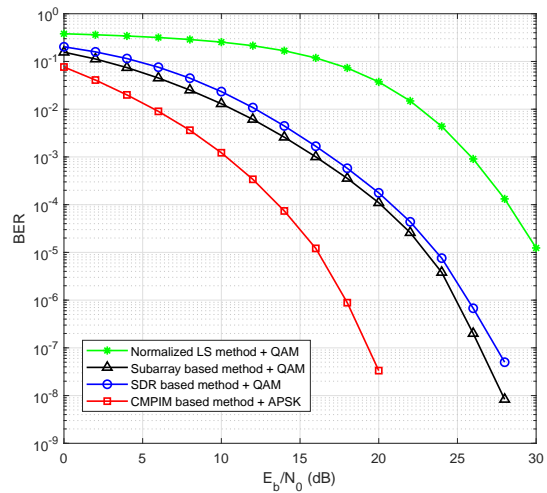}}
\caption{BER performance of directional backscatter communications}\label{BERdir}
\end{center}}
\end{figure}

In \figref{BERdir}, we study the BER performance of our proposed APSK design and reflection pattern design as an integral in directional backscatter communications. For comparison, we adopt three benchmark schemes by combining the traditional constellation design and beam pattern design, i.e., normalized LS method + QAM,  subarray based method + QAM, and SDR based method + QAM, and we set the modulation order as $64$ and receive antenna number as $N_r = 1$. The channel response is  $\mathbf{h} = \mathbf{h}_{LoS} + \mathbf{h}_{NLoS}$, where $\mathbf{h}_{LoS} = \delta \cdot \mathbf{v}(\Psi_x, \Psi_y)$ and $\mathbf{h}_{NLoS} \sim \mathcal{CN}(\mathbf{0}, \sigma^2 \mathbf{I}_{N})$. We further assume that  $\Psi_x \sim U(-0.5, 0.5), \Psi_y \sim U(-0.25, 0.25)$, the sum of reflection loss and propagation loss of each channel realization is $-10\log_{10} |\delta|^2  = 20$dB, and the strength of Non-Line-of-Sight (NLoS) component is $10$dB less than Line-of-Sight (LoS) component, i.e., $10\log_{10} \sigma^2 = 10\log_{10} |\delta|^2 - 10 = -30$dB. By averaging over 10000 channel realizations, we obtain the BER performance as in \figref{BERdir}. It can be seen that, when BER is $10^{-5}$, our proposed design is $7$dB better than the best benchmark scheme, i.e., subarray based method + QAM. Recall that in \figref{BER} the performance enhancement (in terms of BER) of 64-APSK over 64-QAM in AWGN is $1.5$dB, which indicates that the performance improvement (in terms of BER) contributed by our proposed reflection pattern design is $5.5$dB.

\section{Conclusion}
In  this paper, we have studied the design of power-efficient higher-order constellation and reflection pattern under the amplitude constraint.  For constellation design, we adopt the APSK constellation and propose to optimize the ring number, ring radius, and inter-ring phase difference of APSK.   For reflection pattern design, we propose to design the passive beamforming vector by solving a max-min optimization problem under constant modulus constraint, and a constant-modulus power iteration method is proposed to optimize the objective function in each iteration.  Numerical results show that the proposed APSK constellation design and reflection pattern design outperform the existing modulation and beam pattern design in programmable metasurface enabled backscatter communications.
\begin{appendices}

\section{Proof of Proposition 1}
Rewrite the term  $\frac{  k_l}{N_l}-  \frac{  k_{l+1} }{N_{l+1}} $  in \eqref{SetA} as
\begin{align}
 & \; \frac{  k_l}{N_l}-  \frac{  k_{l+1} }{N_{l+1}}  \notag \\
 =& \frac{1}{{\rm lcm}(N_l, N_{l+1})} \underbrace{\left( \frac{N_{l+1}k_l }{{\rm gcd}(N_l, N_{l+1})}  -\frac{N_{l}k_{l+1}}{{\rm gcd}(N_l, N_{l+1})}    \right)}_{\Gamma(k_l, k_{l+1})}
\end{align}
Without loss of generality, we set $\Gamma=c$, with $c$ being an integer and $ 0\leq c < N_l N_{l+1}$, and relax the range of $k_l$ and $k_{l+1}$, i.e.,
\begin{align} \label{Diophantine}
 \frac{N_{l+1}k_l }{{\rm gcd}(N_l, N_{l+1})}  -\frac{N_{l}k_{l+1}}{{\rm gcd}(N_l, N_{l+1})} = c
\end{align}
 Since the variables $k_l$, $k_{l+1}$ and  the parameters $N_{l}$, $N_{l+1}$ are integers,  \eqref{Diophantine} is a linear Diophantine equation. The  greatest common divisor of the coefficients $\frac{N_{l+1} }{{\rm gcd}(N_l, N_{l+1})}$ and $ -\frac{N_{l} }{{\rm gcd}(N_l, N_{l+1})}$ is $1$, and $c$ in \eqref{Diophantine}  is a multiple of  the greatest common divisor. According to the property of  linear Diophantine equation \cite{Diophantine},  \eqref{Diophantine} must have a solution $(\tilde{k}_l, \tilde{k}_{l+1})$. By setting
\begin{align}
\begin{split}
k_l &= {\rm mod}(\tilde{k}_l, N_l) \\
k_{l+1} &=  {\rm mod}(\tilde{k}_{l+1}, N_{l+1})
\end{split}
\end{align}
we have
\begin{align}
&2 \pi\left(\frac{ k_l }{N_l}-  \frac{ k_{l+1}  }{N_{l+1}}\right) = \frac{2 \pi \cdot c }{{\rm lcm}(N_l, N_{l+1})}+ \gamma \cdot 2\pi
\end{align}
where  $\gamma$  is an integer.

Therefore, the feasible region of the term $\cos  \phi_{l, l+1}$  is written as
\begin{align} \label{SetA}
 \Big\{ \cos \Big(\frac{2 \pi \cdot c }{{\rm lcm}(N_l, N_{l+1})} + \Delta\omega_{l, l+1} \Big),
  0\leq c < N_l N_{l+1}\Big\}
\end{align}
The set of \eqref{SetA}  consists of $N_l N_{l+1}$ discrete samplings of the cosine function with the sample interval $\frac{2 \pi }{{\rm lcm}(N_l, N_{l+1})}$. Due to the cyclic  property,  the range of $ \Delta\omega_{l, l+1}$ can be narrowed down to $ \Delta\omega_{l, l+1}\in \big[0, \frac{2\pi}{{\rm lcm}(N_l, N_{l+1})}\big)$.  When $\Delta\omega_{l, l+1}=0$, the largest element in the set of \eqref{SetA} is $\cos \big(\frac{ 0 }{{\rm lcm}(N_l, N_{l+1})}\big)=1$ and the second largest element is $\cos \big(\frac{2 \pi }{{\rm lcm}(N_l, N_{l+1})}\big)$. It is easy to find that when $\Delta\omega_{l, l+1}= \frac{ \pi }{{\rm lcm}(N_l, N_{l+1})} $, the largest element and second largest element become equal. Considering the cyclic property, the optimal phase difference  is represented as
\begin{align}
\Delta\omega^*_{l, l+1}= \frac{ (1+2\nu) \pi }{{\rm lcm}(N_l, N_{l+1})} , \quad {\rm \nu\; is \; an \; integer}
\end{align}
and the corresponding minimum angle is
\begin{align}
\phi^*_{l, l+1}  = \frac{\pi}{\rm lcm(N_l, N_{l+1})}
\end{align}

\section{Proof of Proposition 2}

When $[\Psi_x^L, \Psi_x^U) = [-1, 1)$, we have
\begin{subequations}
\begin{align}
\mathbf{V}_{\Psi_x}(\ell, \kappa) &  = \frac{j 2\sin(\pi (\ell-\kappa) )}{j (\ell-\kappa) \pi } = 0,\;\; \ell\neq \kappa   \\
\mathbf{V}_{\Psi_x}(\ell, \kappa) &= 2,\;\; \ell = \kappa
    \end{align}
\end{subequations}
Namely, $\mathbf{V}_{\Psi_x} = 2\mathbf{I}_{N_x}$. Similarly, when $[\Psi_y^L, \Psi_y^U) = [-1, 1)$, we have  $\mathbf{V}_{\Psi_y} = 2\mathbf{I}_{N_y}$.

Therefore,
\begin{align}
\mathbf{V}_{\mathcal{D}_{\Psi}} = \mathbf{V}_{\Psi_x} \otimes \mathbf{V}_{\Psi_y} = 4\mathbf{I}_{N_x N_y}
\end{align}
and
\begin{align}
P_{\mathcal{D}_{\Psi}} = 4 \mathbf{f}^H \mathbf{f}
\end{align}

\section{Proof of Proposition 3}

\begin{proof}
Firstly, we prove that
\begin{align} \label{Proof3}
|{\mathbf{f}_x^{(i+1)}}^H  \mathbf{M}_{\Psi_x} \mathbf{f}_x^{(i)}|  \geq {\mathbf{f}_x^{(i)}}^H  \mathbf{M}_{\Psi_x} \mathbf{f}_x^{(i)}
\end{align}
The left-hand side term \eqref{Proof3} satisfies
\begin{subequations}   \label{App2a}
\begin{align}
     & |{\mathbf{f}_x^{(i+1)}}^H  \mathbf{M}_{\Psi_x} \mathbf{f}_x^{(i)}|  \notag  \\
    =&   |{\mathbf{f}_x^{(i+1)}}^H (\mathbf{M}_{\Psi_x} + \frac{1}{\delta}\mathbf{I} )\mathbf{f}_x^{(i)} -  \frac{1}{\delta}{\mathbf{f}_x^{(i+1)}}^H  \mathbf{f}_x^{(i)}| \notag \\
    \geq &\frac{1}{\delta} |{\mathbf{f}_x^{(i+1)}}^H \mathbf{f}_{x, temp}^{(i)}|  -  \frac{1}{\delta}|{\mathbf{f}_x^{(i+1)}}^H  \mathbf{f}_x^{(i)}| \label{App2a1} \\
    \geq & {\frac{1}{\delta}  {\mathbf{f}_x^{(i+1)}}^H \mathbf{f}_{x, temp}^{(i)}}   -  \frac{N_x}{\delta} \label{App2a2} \\
    = & \frac{1}{\delta}  \| \mathbf{f}_{x, temp}^{(i)} \|_1    -  \frac{N_x}{\delta} \label{App2a3}
\end{align}
\end{subequations}
where $\| \cdot \|_1$ is $\ell_1$ norm.  \eqref{App2a1} is obtained due to triangle inequality,  \eqref{App2a2} is obtained due to $|{\mathbf{f}_x^{(i+1)}}^H  \mathbf{f}_x^{(i)}|\leq N_x$, and \eqref{App2a3} is obtained as $\mathbf{f}_x^{(i+1)}  = \left[ \frac{\mathbf{f}_{x, temp}^{(i)}(1)}{|\mathbf{f}_{x, temp}^{(i)}(1)|},\;\cdots, \;   \frac{\mathbf{f}_{x, temp}^{(i)}(N_x)}{|\mathbf{f}_{x, temp}^{(i)}(N_x)|} \right]^T $,

The right hand side term of \eqref{Proof3} satisfies
\begin{subequations} \label{App2b}
\begin{align}
&{\mathbf{f}_x^{(i)}}^H  \mathbf{M}_{\Psi_x} \mathbf{f}_x^{(i)} \notag \\
=& {\mathbf{f}_x^{(i)}}^H  (\mathbf{M}_{\Psi_x} + \frac{1}{\delta} \mathbf{I} ) \mathbf{f}_x^{(i)} - \frac{1}{\delta} {\mathbf{f}_x^{(i)}}^H {\mathbf{f}_x^{(i)}} \notag \\
=& \underbrace{\frac{1}{\delta}{\mathbf{f}_x^{(i)}}^H  \mathbf{f}_{x, temp}^{(i)}}_{real \; number} - \frac{N_x}{\delta} \label{App2b1} \\
\leq & \frac{1}{\delta}  \| \mathbf{f}_{x, temp}^{(i)} \|_1    -  \frac{N_x}{\delta} \label{App2b2}
\end{align}
\end{subequations}
Since $\mathbf{M}_{\Psi_x} + \frac{1}{\delta} \mathbf{I} $ is a positive semi-definite matrix,  $\frac{1}{\delta}{\mathbf{f}_x^{(i)}}^H  \mathbf{f}_{x, temp}^{(i)}$ is a real number.  The equality of \eqref{App2b2} holds if and only if $\mathbf{f}_x^{(i+1)} = \mathbf{f}_x^{(i)}$. Combining \eqref{App2a} and \eqref{App2b},  \eqref{Proof3} is obtained.

According to the Cauchy-Schwarz inequality, we have
\begin{align}
|{\mathbf{f}_x^{(i+1)}}^H  \mathbf{M}_{\Psi_x} \mathbf{f}_x^{(i)}|^2  \leq &  \left({\mathbf{f}_x^{(i)}}^H  \mathbf{M}_{\Psi_x} \mathbf{f}_x^{(i)} \right) \cdot  \notag \\
& \left( {\mathbf{f}_x^{(i+1)}}^H  \mathbf{M}_{\Psi_x} \mathbf{f}_x^{(i+1)} \right) \label{CS}
\end{align}
Based on  \eqref{Proof3} and \eqref{CS}, we have
\begin{align}
       {\mathbf{f}_x^{(i+1)}}^H \mathbf{M}_{\Psi_x}   \mathbf{f}_x^{(i+1)} \geq {\mathbf{f}_x^{(i)}}^H \mathbf{M}_{\Psi_x} \mathbf{f}_x^{(i)}
\end{align}
and the equality holds if and only if  $\mathbf{f}_x^{(i+1)}= \mathbf{f}_x^{(i)}$.

\end{proof}

\end{appendices}

\bibliographystyle{IEEEtran}%

\bibliography{bibfile}

\end{document}